\documentclass[12pt, draftclsnofoot, journal, onecolumn]{IEEEtran}
\usepackage[dvips]{graphicx}
\usepackage{amsmath, amsthm, amssymb, amsfonts, mathtools}
\usepackage{cases, multirow, cool}
\usepackage{cite, url}
\usepackage{xargs}                      
\usepackage{xcolor}  
\usepackage{algorithm, algpseudocode}
\usepackage{lipsum, courier}
\usepackage{enumerate}
\usepackage{afterpage}

\makeatletter
\newcommand*\Bigcdot{\mathpalette\Bigcdot@{.5}}
\newcommand*\Bigcdot@[2]{\mathbin{\vcenter{\hbox{\scalebox{#2}{$\m@th#1\bullet$}}}}}
\makeatother

\long\def\comment#1{}

\newtheorem{thm}{Theorem}
\newtheorem{lemm}{Lemma}

\newtheorem{remark}{Remark}

\def\figref#1{Fig.~\ref{#1}}
\def\be{\begin{equation} }
\def\ee{\end{equation}}

\newlength{\myimagewidth}
\newlength{\myimageheight}

\title{A Comprehensive Analysis of 5G Heterogeneous Cellular Systems operating over $\kappa$-$\mu$ Shadowed Fading Channels}

\begin{document}
\setlength{\myimagewidth}{\dimexpr\linewidth/2-1em\relax}
\setlength{\myimageheight}{85 mm}

\author{
   \IEEEauthorblockN{Young Jin Chun, Simon L. Cotton, Harpreet S. Dhillon,\\ F. Javier Lopez-Martinez, Jos\'e F. Paris, and Seong Ki Yoo}
   \thanks{
Y. J. Chun, S. L. Cotton and S. K. Yoo are with the Wireless Communications Laboratory, ECIT Institute, Queens University Belfast, United Kingdom.
   H. S. Dhillon is with Wireless@VT, Department of Electrical and Computer Engineering, Virginia Tech, Blacksburg, VA, USA.
   F. J. Lopez-Martinez and J. F. Paris are with Departmento de Ingenier\'ıa de Comunicaciones, Universidad de Malaga, Malaga, Spain. 
   (Email: \{Y.Chun, simon.cotton, syoo02\}@qub.ac.uk, hdhillon@vt.edu, \{fjlopezm, paris\}@ic.uma.es). 
   }
   }

\maketitle

\vspace{-12mm}
{
\it \bf Abstract}-
{
Emerging cellular technologies such as those proposed for use in 5G communications will accommodate a wide range of usage scenarios with diverse link requirements. This will include the necessity to operate over a versatile set of wireless channels ranging from indoor to outdoor, from line-of-sight (LOS) to non-LOS, and from circularly symmetric scattering to environments which promote the clustering of scattered multipath waves. Unfortunately, many of the conventional fading models lack the flexibility to account for such disparate signal propagation mechanisms. To bridge the gap between theory and practical channels, we consider $\kappa$-$\mu$ shadowed fading, which contains every linear fading models proposed in the open literature as special cases. In particular, we propose an analytic framework to evaluate the average of an arbitrary function of the SINR over $\kappa$-$\mu$ shadowed fading channels by using a simplified orthogonal expression with tools from stochastic geometry. Using the proposed method, we evaluate the spectral efficiency, moments of the SINR, and outage probability of a $K$-tier HetNet with $K$ classes of BSs, differing in terms of the transmit power, BS density, shadowing and fading. Building upon these results, we provide important new insights into the network performance of these emerging wireless applications while considering a diverse range of fading conditions and link qualities.
}

\clearpage



\section{Introduction}

To meet the ever-increasing demand for data on the move, telecommunications industries, as well as global standardization entities, are actively driving the research and development of the fifth generation (5G) of wireless communications. It is forecast that this new networking paradigm will provide $1000$ fold gains in capacity over the next decade and data rates exceeding $10$ Gigabit/s while achieving latencies of less than $1$ millisecond \cite{5GPPP2016,Nokia2014}. To make this possible, 5G communications will utilize densely deployed small cells to achieve high spectral efficiency while harnessing all available spectrum resources, including opportunities offered by millimeter-wave frequencies. Key to the successful operation of 5G communications will be the unification of dissimilar networking technologies. This will create a diverse range of link requirements and the necessity for wireless devices to operate over a versatile set of channels ranging from indoor to outdoor, from line-of-sight (LOS) to non-LOS (NLOS), and from circularly symmetric scattering to those which promote the clustering of scattered multipath waves.

A range of tools developed within the framework of stochastic geometry have been used to capture the irregularity and heterogeneity of 5G wireless networks with considerable success. Specifically, stochastic geometry assumes that the locations of all wireless nodes are endowed with a spatial point process \cite{Haenggi2013}. Such an approach captures the topological randomness in the network geometry, allows the use of well-established mathematical tools, offers high analytical flexibility and achieves an accurate performance evaluation \cite{Andrews2010}. A common assumption made within this scheme is that the nodes are distributed according to a Poisson point process (PPP). Using this supposition, the probability density function (PDF) of the aggregate interference and the outage probability were analyzed for cellular network in \cite{Mathar1995, Andrews2011}, which were generalized to the case of heterogeneous cellular networks (HetNets) in \cite{Jo2012,Mukherjee2012, Dhillon2012, Chun2015,Saha2016}\footnote{The aforementioned results represent only a subset of the related studies in stochastic geometry. The interested reader is directed to the work presented in \cite{Andrews2016,Elsawy2013,Haenggi2009} and the references therein for a more detailed overview of stochastic geometry.}.


Much of the existing published work on stochastic geometry has focused on the Rayleigh distribution as the small-scale fading model, owing to its simplicity and tractability. Several approaches have been proposed to derive the signal-to-noise-plus-interference ratio (SINR) distributions for general fading environments. For instance, in \cite{Blaszczyszyn2013,Keeler2013, Madhusudhanan2014a, Dhillon2014, Zhang2014} the conversion method, which is based on displacement theorem, was used. This method treats the channel randomness as a perturbation in the location of the transmitter and transforms the original network with arbitrary fading into an equivalent network without fading. Although the conversion method can be applied to any fading distribution, it is more tractable for handling large-scale shadowing effects. Specifically, if one applies the conversion method to small-scale fading, the resulting equivalent model will have no fading, thereby the Laplace transform-based approach can not be utilized.  An alternative approach to address general fading scenarios uses the series representation method \cite{Peng2014,Tanbourgi2014}. This approach expresses the interference functionals as an infinite series of higher order derivative terms given by the Laplace transform of the interference power. While the series representation method provides a tractable alternate for handling general fading, it often leads to situations where it is difficult to derive closed form expressions. Numerically evaluating a higher order derivative is also complex and prone to floating-point rounding errors \cite{Miel1985}. 

Aside from the small-scale fading, random shadowing due to obstacles in the local environment or human body movements (in the case of user equipments) can impact link performance by causing fluctuations in the received signal. Shadowing affects the transmission performance which will be especially pertinent in a dense network or millimeter-wave links. Hence, the combined effect of small-scale and shadowed fading needs to be properly addressed in 5G communications design. In this respect, composite channel models have been proposed in \cite{Ho1993,Abdi1998,Abdi2003,Paris2014,6953146}. In \cite{Ho1993}, the shadowed Nakagami fading distribution was first proposed by combining Nakagami-\textit{m} multipath fading and lognormal distributed shadowing. Later, \cite{Abdi1998} introduced the generalized-\textit{K} model by approximating the shadowing model in \cite{Ho1993} using the gamma distribution to improve analytical tractability. Traditional composite channel models (referred to as \textit{multiplicative shadow fading model}) assume that the shadowing affects the dominant components and the scattered waves equally, whereas, in practice, the shadowing often occurs only on the dominant components, which gives rise to a different kind of model and is referred to as \textit{LOS shadow fading model}. To model shadowing in the LOS channels, \cite{Abdi2003} proposed the Rician shadowed fading model by assuming a Rician distribution for the multipath fading and Nakagami-\textit{m} distribution for the LOS shadowing. More recently, \cite{Paris2014,6953146} proposed $\kappa$-$\mu$ shadowed fading model by assuming $\kappa$-$\mu$ multipath fading with shadowing of the dominant component.

The $\kappa$-$\mu$ shadowed fading model is an attractive proposition, not just due to its excellent fit to the fading observed in a range of real-world applications (e.g. device-to-device \cite{6953146}, underwater acoustic \cite{Sanchez2014}, body-centric fading channels \cite{Cotton2014}) but also its extreme versatility. More precisely, it is able to account for most of the popular fading distributions utilized in the literature.

Motivated by the comprehensive nature of the $\kappa$-$\mu$ shadowed fading model, we use it along with a stochastic geometric framework to derive the downlink SINR distribution of a typical user in a $K$-tier HetNet with $K$ classes of BSs, differing in terms of the transmit power, BS density, shadowing and fading characteristics. We evaluate the average of an arbitrary function of the SINR, which can be easily applied to other network models. For instance, it may be utilized to evaluate any performance measure that can be represented as a function of SINR, \textit{e.g.}, the spectral efficiency, outage probability, moments of the SINR, and error probability. 

The main contributions of this paper may be summarized as follows.
\begin{enumerate}
  \item The main difficulty with incorporating generalized small-scale fading models into stochastic geometry framework is the lack of tractability in expressing the PDF of the interference. In general, it is more convenient to express the metrics of interest in terms of the Laplace transform of the interference. Nonetheless, this presents significant challenges when extending the analyses from Rayleigh fading to the more general fading models. We overcome this problem by analyzing the Laplace transform of the interference over $\kappa$-$\mu$ shadowed channels to characterize the distribution of the interference from cellular user equipment (UE). It is worth highlighting that this model encompasses majority of the fading models proposed in the literature as special cases, including Rayleigh, Rician, Nakagami-\textit{m}, Nakagami-\textit{q}, One-sided Gaussian, $\kappa$-$\mu$, $\eta$-$\mu$, and Rician shadowed distribution to name but a few.

  \item We use tools from stochastic geometry to evaluate the distribution of the SINR, coverage probability and average rate for $\kappa$-$\mu$ shadowed fading. We also propose a numerically efficient method to calculate the average of an arbitrary function of the SINR. 
  
  \item We present numerical simulation results which provide useful insights into the performance of cellular networks for different fading conditions. In particular, we observe the trade-off relation between the rate and average SINR based on the channel parameters, such as the intensity of dominant signal components, the number of scattering clusters, and shadowing effect. This information will be of paramount importance to those responsible for designing future 5G network infrastructure to ensure that adequate service can be provided.
\end{enumerate}

This paper is organized as follows. In Section II, the system model and assumptions are introduced. We then apply an orthogonal expansion to $\kappa$-$\mu$ shadowed PDF in Section III, characterize the interference distribution in Section IV, and introduced a novel analytical framework in Section V. Following this, in Section VI, we present numerical and simulation results to validate the analysis. Finally, Section VII concludes this paper.

\section{System Model}

\subsection{Network Model}

We consider the downlink of a $K$-tier HetNet where randomly distributed small-cell BSs, such as pico or femtocell BSs, are overlaid on a network of macrocell BSs. The BSs of each tier may differ in terms of transmit power and spatial density. The locations of the $k$-th tier BSs are modeled by an independent, homogeneous PPP $\Phi_k$ with density $\lambda_k$ and the union of $K$ point processes constitutes the $K$-tier HetNet $\displaystyle\Phi = \mathop{\cup}_{k \in \mathcal{K}} \Phi_k$ where $\mathcal{K} = \{1,2, \ldots, K\}$. The locations of the UEs are modeled by a homogeneous PPP $\Phi^{(u)}$ with density $\lambda^{(u)}$ that is independent of $\Phi$. Orthogonal multiple access is employed at each cell by allocating mutually orthogonal resource blocks to each UE, implying no intra-cell interference within a cell. Without loss of generality, we assume that a typical UE is located at the origin and each BS has an infinitely backlogged queue. The received power at a typical UE from a $k$-th tier BS $x_k \in \Phi_k$\footnote{$x_k$ denotes both the node and the coordinates of the BS.} is given by 
\begin{equation}
  \begin{split}
    P_{x_k} = P_k H_{x_k} \left( \tau \|x_k\|^{-\alpha}\right) 
    = P_k h_{x_k} \chi_{x_k} \left( \tau \|x_k\|^{-\alpha}\right),
  \end{split}
  \label{eq_cc_001}
\end{equation}
where a multiplicative channel model $H_{x_k} = h _{x_k} \chi_{x_k}$ with large-scale shadowing $\chi$ and small-scale fading $h$ is utilized in the second equality, $P_k$ is the transmit power of the $k$-th tier BS, $\alpha$ is the path-loss exponent, and $\tau$ is the path-loss intercept at a link-length $\|x\| = 1$.

\subsection{Cell Association Policy}

We assume a general cell association model where all BSs allow open access and each UE connects to the BS that provides the highest long-term received power (LRP)\footnote{The interested reader is referred to \cite{Zhang2014,Jo2012,Dhillon2014} for a detailed description on the long-term association scheme.} without small-scale fading as written below 
 \begin{equation}
   \begin{split}
     &\text{Typical UE associates to a $k$-th tier BS }x_k^{\ast} \in \Phi_k\\
  \leftrightarrow ~ &x_k^{\ast} = 
  \mathop{\mathrm{arg ~max}}_{j \in \mathcal{K},~x \in \Phi_j} P_j \chi_j \|x\|^{-\alpha} = 
  \mathop{\mathrm{arg ~max}}_{j \in \mathcal{K},~y \in \Phi^{(e)}_j}  P_j \|y\|^{-\alpha},
   \end{split}
   \label{eq_cc_001-ext2}
 \end{equation}
where a change of variable, \textit{i.e.}, $y = \chi_j^{-\frac{1}{\alpha}} x$, is applied in the last equality. For a single tier network, (\ref{eq_cc_001-ext2}) is equivalent to connecting with the closest BS. 

Due to the displacement theorem \cite[Lemma 1]{Dhillon2014}, the mapping between $x$ and $y$ converts a PPP $\Phi_j = \{x\}$ with density $\lambda_j$ into a new homogeneous PPP $\Phi_j^{(e)} = \{y\}$ with density $\lambda_j^{(e)} = \lambda_j \mathbb{E}\left[ \chi_j^{\delta}\right]$ where $\delta = \frac{2}{\alpha}$. Thereby, the original network model $\displaystyle\Phi$ with large-scale shadowing $\chi$ can be equivalently expressed as the network $\displaystyle\Phi^{(e)} = \mathop{\cup}_{j \in \mathcal{K}} \Phi_j^{(e)}$ without a large-scale shadowing where the effect of large-scale shadowing is now incorporated through an appropriate scaling in the density $\lambda_j \rightarrow \lambda_j^{(e)}$. Given that the serving BS belongs to the $k$-th tier, the SINR at a typical UE can be formulated as follows.
 \begin{equation}
   \begin{split}
     \mathrm{SINR}_k = 
     \frac{P_k \chi_{x_k^{\ast}} h _{x_k^{\ast}}  \|x_k^{\ast}\|^{-\alpha}}
     {N + \sum_{j \in \mathcal{K}} \sum_{x \in \Phi_j\backslash\{ x_k^{\ast}\}}  P_j \chi_{x} h_{x}  \|x\|^{-\alpha}}
     = 
     \frac{h _{y_k^{\ast}}  \|y_k^{\ast}\|^{-\alpha}}
     {\hat{N} + \sum_{j \in \mathcal{K}} \sum_{y \in \Phi_j^{(e)}\backslash\{ y_k^{\ast}\}} \hat{P}_j h_{y}  \|y\|^{-\alpha}},
   \end{split}
   \label{eq_cc_002}
 \end{equation}
where $x_k^{\ast}$ represents the location of the associated $k$-tier BS, $\Phi\backslash\{ x_k^{\ast}\}$ denote the set of interfering BSs, $\hat{P}_j = \frac{P_j}{P_k}$ represents the ratio between the transmit power of the interfering and serving BS and $\hat{N} = \frac{N}{P_k} = \frac{N_0 W}{\tau P_k}$ is determined by the noise power spectral density $N_0$, bandwidth $W$, transmit power of the associated BS $P_k$, and the reference path-loss $\tau$ at a unit distance. We denote $I \triangleq \sum_{j \in \mathcal{K}} \sum_{y \in \Phi_j^{(e)}\backslash\{ y_k^{\ast}\}} \hat{P}_j h_{y}  \|y\|^{-\alpha}$ as the aggregate interference normalized by the transmit power of the serving BS. Since the cell association policy in (\ref{eq_cc_001-ext2}) is independent of the small-scale fading distribution $h$, the probability that a typical UE connects to the $k$-th tier BS, denoted as $\mathcal{P}_k$, and the PDF of the link length $\|y_k^{\ast}\|$ can be evaluated as below 
 \begin{equation}
   \begin{split}
   \mathcal{P}_k &= \frac{\lambda_k \mathbb{E}\left[ \chi_k^{\delta}\right]}
   {\sum_{j \in \mathcal{K}} \lambda_j \mathbb{E}\left[ \chi_j^{\delta}\right] \hat{P}_j^{\delta}},
   \quad
   f_{\|y_k^{\ast}\|}(r) = \frac{2 \pi \lambda_k \mathbb{E}\left[ \chi_k^{\delta}\right]}{\mathcal{P}_k} r 
   \exp\left[ 
   -\sum_{j \in \mathcal{K}} \pi r^2 \lambda_j \mathbb{E}\left[ \chi_j^{\delta}\right] \hat{P}_j^{\delta} \right],
   \end{split}
   \label{eq_cc_001-ext3}
 \end{equation}
where  $\delta = \frac{2}{\alpha}$ and (\ref{eq_cc_001-ext3}) follows directly from \cite[Lemma 1]{Jo2012} and \cite[Lemma 2]{Dhillon2014}.

\subsection{Channel Model}

Due to the wide range of use cases provisioned for 5G communications, conventional cellular channel models which typically only consider a single source of shadowing (e.g. large-scale shadowing) are unlikely to be general enough. In reality, it is probable that cellular applications will encounter multiple independent types of shadowing which may or may not occur concurrently. For example in the downlink scenario, the signal transmitted from the BS to the UE will undergo two key types of shadowing, the first of which is large-scale shadowing, denoted here by $\chi$, which is induced  due to large terrestrial objects \textit{e.g.} buildings or hills, which can cause a random fluctuation in the total signal power. In cellular networks, the BSs are usually positioned in elevated locations and are typically free from surrounding clutter. However, UEs are most often operated at lower levels and the LOS signal path is often obscured by local obstacles including the user's body itself. Therefore we consider a second type of shadowing which affects (i.e. randomly fluctuates) the dominant signal component. In this contribution, this LOS shadowed small-scale fading is denoted as $h$ and is modeled as a $\kappa$-$\mu$ shadowed random variable \cite{Paris2014, 6953146}. Together, these two independent random processes create an extremely versatile channel model, $H = h\chi$, which can incorporate a wide range of shadowing and fading scenarios.

\subsubsection{Large-Scale Shadowing}
While the analysis presented in this paper is valid for any finite distribution of the large-scale shadowing $\chi$, we limit our investigation to the three most commonly used large-scale shadowing distributions, namely the lognormal, gamma, and inverse-Gaussian distributions \cite{Simon2004}. The corresponding PDF and $j$-th moment of each of the considered distributions are summarized below, where $j$ is a positive real number. 

	\begin{enumerate}[(a)]
		\item Lognormal Shadowing
		 \begin{equation}
		   \begin{split}
		\chi &\sim {LN}(\mu_l, \sigma^2_l)~\text{where }
	   \begin{dcases}
		&f_{\chi}(x) = \frac{\epsilon_0}{\sqrt{2\pi} \sigma_l x} \exp\left[ -\frac{\left( 10\log_{10}x - \mu_l\right)^2}{2 \sigma_l^2}\right],\\
		&\mathbb{E}\left[ \chi^j \right] = \exp\left[ \frac{j\mu_l}{\epsilon_0} + \frac{1}{2}\left( \frac{j \sigma_l}{\epsilon_0} \right)^2\right],\\
		&\text{$\mu_l$ and $\sigma_l$ are expressed in decibels,~ $\epsilon_0 = \frac{10}{\ln(10)}$},
	   \end{dcases}
		   \end{split}
		   \label{eq_cc_002-ext1}
		 \end{equation}
		\item Gamma Shadowing
		 \begin{equation}
		   \begin{split}
		\chi &\sim {Gamma}(k_g, \theta_g)~\text{where }
	   \begin{dcases}
		f_{\chi}(x) &= \frac{1}{\Gamma\left( k_g \right) \theta_g^{k_g}} x^{k_g-1} \exp\left( -\frac{x}{\theta_g}\right),\\
		\mathbb{E}\left[ \chi^j \right] &= \frac{\Gamma(j+k_g) \theta_g^j}{\Gamma(k_g)},~
	    \mathbb{E}\left[ \chi \right] = k_g \theta_g,\\
	   \end{dcases}
		   \end{split}
		   \label{eq_cc_002-ext2}
		 \end{equation}	
		\item Inverse Gaussian Shadowing
		 \begin{equation}
		   \begin{split}
		\chi &\sim {IG}(\mu_{ig}, \lambda_{ig})~\text{where }
	   \begin{dcases}
		f_{\chi}(x) &= \sqrt{\frac{\lambda_{ig}}{2\pi x^3}} \exp\left( -\frac{\lambda_{ig} (x-\mu_{ig})^2}{2 \mu_{ig}^2 x}\right),\\
		\mathbb{E}\left[ \chi^j \right] &= 
		\mathrm{e}^{\frac{\lambda_{ig}}{\mu_{ig}}} \sqrt{\frac{2 \lambda_{ig}}{\pi}}
		\mu_{ig}^{j-\frac{1}{2}} {K}_{\frac{1}{2}-j}\left( \frac{\lambda_{ig}}{\mu_{ig}}\right),   
	   \end{dcases}   
		\end{split}
		   \label{eq_cc_002-ext3}
		 \end{equation}	
	\end{enumerate}
where $K_{n}(z)$ is a modified Bessel function of the second kind. 
The tractability of the gamma and inverse-Gaussian distributions is the reason why they are commonly used in the literature to approximate the lognormal distribution. A one-to-one mapping between the shadowing parameters can be derived by matching the mean and variance of each distribution as illustrated in Fig. \ref{fig:shadowing_parameter}.

\subsubsection{Small-Scale Fading and LOS Shadowing}

The $\kappa$-$\mu$ shadowed distribution is a very flexible model which contains as special cases the majority of the linear fading models proposed in the open literature, including Rayleigh, Rice (Nakagami-$n$), Nakagami-$m$, Hoyt (Nakagami-$q$), One-Sided Gaussian, $\kappa$-$\mu$, $\eta$-$\mu$ and Rician shadowed to name a few \cite{Moreno-Pozas2015} (See Table 1). Because of this generality, the $\kappa$-$\mu$ shadowed fading model can be used to account for small-scale fading which originates due to LOS or non-LOS conditions, multipath clustering with circularly symmetric or elliptical scattering, and power imbalance between the in-phase and quadrature signal components.

The channel coefficient $h$ of a $\kappa$-$\mu$ shadowed fading channel can be expressed in terms of the in-phase and quadrature components of the fading signal as follows
 \begin{equation}
   \begin{split}
		h = \sum_{i=1}^{\mu}\left[ \left( X_i + \xi p_i \right)^2 + \left( Y_i + \xi q_i \right)^2\right],
   \end{split}
   \label{eq_cc_003}
 \end{equation}
 where $\mu$ is the number of the multipath clusters\footnote{Note that $\mu$ is initially assumed to be a natural number, however this restriction is relaxed to allow $\mu$ to assume any positive real value.}, $X_i$ and $Y_i$ are mutually independent Gaussian random variables with
 \begin{equation}
   \begin{split}
		\mathbb{E}\left[ X_i \right] = \mathbb{E}\left[ Y_i \right] = 0, \quad
		\mathbb{E}\left[ X_i^2 \right] = \mathbb{E}\left[ Y_i^2 \right] = \sigma^2, \quad
		d^2 = \sum_{i = 1}^{\mu}\left( p_i^2 + q_i^2\right),
   \end{split}
   \label{eq_cc_004}
 \end{equation}
 and $\xi$ is a Nakagami-$m$ distributed shadowing perturbation with $\mathbb{E}\left[\xi^2\right] = 1$.

In the following, we summarize the statistics of the $\kappa$-$\mu$ shadowed fading model which will be used in the network performance analysis conducted here. The PDF, $j$-th moment, and Laplace transform of $h$ for the $\kappa$-$\mu$ shadowed channel are respectively given by \cite{Paris2014}
 \begin{equation}
   \begin{split}
    f_h(x) &= \frac{\theta_1^{m-\mu} x^{\mu-1}}{\theta_2^m \Gamma(\mu)} \exp\left( -\frac{x}{\theta_1}\right)
    \Hypergeometric{1}{1}{m}{\mu}{\frac{\theta_2 - \theta_1}{\theta_1 \theta_2}x}, \\
    \mathbb{E}\left[ h^j \right] &= \frac{\theta_1^{m-\mu}  \Gamma\left( \mu + j \right)}{\theta_2^{m - \mu- j} \Gamma\left( \mu \right)}
    \Hypergeometric{2}{1}{\mu - m, \mu + j}{\mu}{-\frac{\mu \kappa}{m}}, \\
    \mathcal{L}_h(s) &= \mathbb{E}\left[ \exp\left( -s h\right) \right] =
    \left( 1 + \theta_1 s\right)^{m-\mu} \left( 1 + \theta_2 s\right)^{-m},
   \end{split}
   \label{eq_cc_005}
 \end{equation}
where $\bar{h} = \mathbb{E}[h]$, $\theta_1 = \frac{\bar{h}}{\mu (1 + \kappa)}$, $\theta_2 = \frac{(\mu \kappa + m) \bar{h}}{\mu (1 + \kappa) m}$, $\kappa$, $\mu$, $m$ and $j$ are positive real-valued constants, $\Gamma(t)$ is the gamma function defined in (\ref{eq_app_I-003}), and $\Hypergeometric{1}{1}{a}{b}{x}$ is the confluent hypergeometric function.
The $j$-th moment of a $\kappa$-$\mu$ shadowed distributed random variable is derived as below
 \begin{equation}
   \begin{split}
  \mathbb{E}\left[ h^j \right] 
  &= \frac{\theta_1^{m-\mu} }{\theta_2^m \Gamma(\mu)} \left( \frac{\theta_1 \theta_2}{\theta_2 - \theta_1} \right)^{\mu + j} \int_{0}^{\infty} t^{\mu+j-1} \mathrm{e}^{ -\frac{m t}{\mu \kappa}}
  \Hypergeometric{1}{1}{\mu - m}{\mu}{-t} \mathrm{d}t\\
  &= \frac{\theta_1^{m-\mu}  \Gamma\left( \mu + j \right)}{\theta_2^{m - \mu- j} \Gamma\left( \mu \right)}
    \Hypergeometric{2}{1}{\mu - m, \mu + j}{\mu}{-\frac{\mu \kappa}{m}},
   \end{split}
   \label{eq_cc_006}
 \end{equation}
where $j$ is a positive real number, a change of variables, \textit{i.e.}, $\frac{\theta_2 - \theta_1}{\theta_1 \theta_2}x \rightarrow t$, with (\ref{eq_app_I-002-a-ext1}) are used in the first equality, then (\ref{eq_app_I-001-b}) is applied in the last equality.

Physically, $\kappa = \frac{d^2}{2\mu \sigma^2}$ represents the ratio between the total power of the dominant components and the total power of the scattered waves, $\mu$ denotes the real-valued extension of the number of multipath clusters, and $m$ indicates the amount of shadowed perturbation in the dominant component as illustrated in Fig. \ref{fig:channel_parameter}. Since the Laplace transform of the Nakagami-\textit{m} distribution converges to $\lim_{m \rightarrow \infty} \mathcal{L}_h(s) = \lim_{m \rightarrow \infty} (1 + s \bar{h}/m)^{-m} = \mathrm{e}^{-s \bar{h}}$, the dominant component becomes increasingly deterministic as $m \rightarrow \infty$. Hence, a $\kappa$-$\mu$ shadowed fading channel where $m \rightarrow \infty$ has a constant dominant power and is therefore equivalent to a $\kappa$-$\mu$ faded channel.

\subsubsection{Combined Large-Scale Shadowing, Small-Scale Fading and LOS Shadowing}

Since the $\kappa$-$\mu$ shadowed fading model includes small-scale fading and LOS shadowed fading as special cases, the proposed channel model $H = h \chi$ can be used to represent four different classes of fading environment as illustrated in Fig. \ref{fig:channel_parameter2}; namely 1) small-scale fading only if $\chi$ is constant, 2) small-scale fading with LOS shadowed fading only if $h$ is either Rician shadowed or $\kappa$-$\mu$ shadowed and $\chi$ is constant, 3) traditional composite fading/shadowing model if $h$ is the result of small-scale fading only with randomly distributed $\chi$, and 4) double shadowed fading conditions if $h$ is the result of small-scale and LOS shadowed fading and $\chi$ is a random variable.

\section{Laguerre Polynomial Series Expansion of the $\kappa$-$\mu$ Shadowed Distribution}

As we can see from (\ref{eq_cc_005}), the $\kappa$-$\mu$ shadowed distribution includes the hypergeometric function which often leads to computationally complex performance evaluation. Due to mathematical intractability, limited work has been conducted which considers $\kappa$-$\mu$ shadowed fading in the context of stochastic geometry. Most notably, in \cite{Kumar2015}, the author approximated a $\kappa$-$\mu$ shadowed random variable using a gamma distributed random variable based on second-order moment matching, but the accuracy of this approximation can not be guaranteed for all fading parameters. In \cite{Parthasarathy2016}, the authors analyzed a cellular network over $\kappa$-$\mu$ shadowed fading where they represented the confluent hypergeometric function by its truncated series form, \textit{i.e.}, $\Hypergeometric{1}{1}{a}{b}{x} \simeq \sum_{n=0}^{N}\frac{\Gamma(a+n)\Gamma(b) x^n}{\Gamma(a)\Gamma(b+n) n!}$. Although the series representation converges locally, it is valid only for integer-valued parameters $a$ and $b$, the radius of convergence diverges over different combinations of parameters, and is computationally complex to evaluate. As illustrated in \figref{fig:pdf_compare}, there are noticeable discrepancies between the approximation methods proposed in \cite{Kumar2015} and \cite{Parthasarathy2016} and the exact PDF for several cases, limiting their application\footnote{The approximation accuracy of \cite{Parthasarathy2016} depends on $N$. For a larger $N$, \cite{Parthasarathy2016} may accurately approximate the exact PDF. In contrast, the proposed approach in (\ref{eq_cc_010}) converges rapidly to the exact PDF  even with a small number of terms $N \leq 50$.}.

To overcome this problem, we adopt the generalized Laguerre polynomial expansion proposed in \cite{Chai2009,Abdi2009} that is analogous to the Fourier series: As a Fourier series can represent any PDF in terms of harmonic bases, we use a generalized Laguerre polynomial as an orthogonal base and simplify the PDF and CDF of the $\kappa$-$\mu$ shadowed fading model as given below.
\begin{lemm}
	The PDF and CDF of the channel coefficient $h$ for the $\kappa$-$\mu$ shadowed fading model can be expressed in series expression form as follows
   \begin{align}
	f_h(x) &= \sum_{n=0}^{\infty} \frac{n! C_n L_n^{\mu-1}(x)}{\Gamma(n+\mu)} x^{\mu-1} \exp\left( -x \right)
	= \sum_{n=0}^{\infty} \sum_{i=0}^{n} c_{i, n} ~x^{\mu+i-1} \exp\left( -x \right),
	\label{eq_cc_010}
	\\
		F_h(x) &= \int_{0}^{x} f_h(t) \mathrm{d}t =
	\sum_{n=0}^{\infty} \sum_{i=0}^{n} b_{i, n} ~x^{\mu+i} \exp\left( -x \right) + \frac{\gamma(\mu, x)}{\Gamma(\mu)},
	\label{eq_cc_011}
   \end{align}
where $\kappa$, $\mu$ and $m$ are positive real-valued parameters, $L_n^{\mu-1}(x)$ is the generalized Laguerre polynomial of degree $n$ and order $\mu-1$ at $x$, $0 \leq x < \infty$, 
$\gamma(\mu, x)$ is the lower incomplete gamma function, the coefficients $C_n$, $c_{i, n}$, and $b_{i, n}$ are calculated as written below
 \begin{equation}
   \begin{split}
   &C_n = \sum_{j = 0}^{n} \frac{(-1)^j}{ j!} \binom{n+\mu-1}{n-j} \mathbb{E}\left[ h^j \right], 
   c_{i, n} = \frac{(-1)^i C_n}{\Gamma(\mu + i)} \binom{n}{i}, 
   b_{i, n} = \frac{(-1)^i C_{n+1}}{\Gamma(\mu + i + 1)} \binom{n}{i},
   \end{split}
	\label{eq_cc_012}
 \end{equation}
and $\mathbb{E}\left[ h^j \right]$ is derived in (\ref{eq_cc_006}).
\end{lemm}

\begin{proof}
	See Appendix II.
\end{proof}

\begin{remark}
If $\mu$ and $m$ are positive integers, then by using \cite[Theorem 1]{Lopez-Martinez2016}, the expression in (\ref{eq_cc_010}) can be simplified to a single summation with finite terms as follows.
	 \begin{equation}
	   \begin{split}
	   f_h(x) = 
	   \begin{dcases}
	   \sum_{j=1}^{\mu-m} \frac{A_{1 j} x^{\mu-m-j} \mathrm{e}^{-\frac{x}{\theta_1}}}{\Gamma(\mu-m-j+1) \theta_1^{\mu-m-j+1}}
	    + \sum_{j=1}^{m} \frac{A_{2 j}  x^{m-j} \mathrm{e}^{-\frac{x}{\theta_2}}}{\Gamma(m-j+1) \theta_2^{m-j+1}} &\text{for } m < \mu\\
	   \sum_{j=0}^{m-\mu} \frac{B_{j}  x^{m-j-1} \mathrm{e}^{-\frac{x}{\theta_2}}}{\Gamma(m-j) \theta_2^{m-j}}  &\text{for } m \geq \mu
	   \end{dcases},
	   \end{split}
		\label{eq_cc_011_ext1}
	\end{equation}
	where $A_{1 j}$, $A_{2 j}$, $B_{j}$ are given in \cite[eq (6)]{Lopez-Martinez2016}. (\ref{eq_cc_010}) and (\ref{eq_cc_011_ext1}) imply that $\kappa$-$\mu$ shadowed fading is the result of a linear combination of Gamma distributed random variables, which follows a gamma mixture distribution. To represent the $\kappa$-$\mu$ shadowed fading as a gamma mixture model, double summation with  infinite terms are required for real valued $\mu$ and $m$, whereas for integer valued $\mu$ and $m$, only a single summation with finite terms are necessary.
\end{remark}

\section{Distribution of the Aggregate Interference}

In this section, we calculate the Laplace transform of the aggregate interference for the $\kappa$-$\mu$ shadowed fading channel and characterize the distribution of the interference. The Laplace transform of the aggregate interference is a crucial measure for evaluating the network performance in stochastic geometry based analysis as will be evident from the discussion in Section V.

\begin{lemm}
	Given that a typical UE is associated to the BS $y_k^{\ast}$ located at $\|y_k^{\ast}\| = r$ (or equivalently expressed as $x_k^{\ast}$ using $x = \chi_j^{\frac{1}{\alpha}} y$), the Laplace transform of the aggregate interference over a multiplicative channel with $\kappa$-$\mu$ shadowed fading and large-scale shadowing is calculated as
	 \begin{equation}
	   \begin{split}
	   	\mathcal{L}_{I}(s) 
	   	&= \mathbb{E}\left[ \exp(-s I) \right] = \mathbb{E}\left[ \exp\left(-s \sum_{j \in \mathcal{K}} 
	   	\sum_{y \in \Phi_j^{(e)}\backslash\{ y_k^{\ast}\}}  \hat{P}_j h_{y} \|y\|^{-\alpha}
	   	\right) \right]\\
	   	&= \prod_{j \in \mathcal{K}} \mathcal{L}_{I_j}(s) = 
	   	\exp\left[ -\sum_{j \in \mathcal{K}} 
	   	\pi r^2 \lambda_j \mathbb{E}\left[ \chi_j^{\delta} \right] \hat{P}_j^{\delta} \mathcal{W}_j(z)
	   	\right],
	   \end{split}
		\label{eq_cc_012_ext1}
	 \end{equation}
	 where $\hat{P}_j = \frac{P_j}{P_k}$ is the ratio between the transmit power of the interfering and serving BS, 
	 \begin{equation}
	   \begin{split}
	   \mathcal{W}_j(z) = \frac{\mu ~\theta_1 z}{1-\delta} \left( \frac{\theta_1}{\theta_2} \right)^m \frac{F_2\left( \mu+1; m, 1; \mu, 2-\delta; A, B \right)}{\left( 1+\theta_1 z \right)^{\mu+1}} - \left[ 1 - \frac{\left( 1+\theta_1 z \right)^{m-\mu}}{\left( 1+\theta_2 z\right)^{m}} \right],
	   \end{split}
		\label{eq_cc_012_ext2}
	 \end{equation}
	 for $z = s r^{-\alpha}$, $A = \frac{1-\theta_1/\theta_2}{1+\theta_1 z}$, $B = \frac{\theta_1 z}{1+\theta_1 z}$, $\theta_1 = \frac{\bar{h}}{\mu (1 + \kappa)}$, $\theta_2 = \frac{(\mu \kappa + m) \bar{h}}{\mu (1 + \kappa) m}$ and $\mathbb{E}\left[ \chi_j^{\delta} \right]$ is given by (\ref{eq_cc_002-ext1})-(\ref{eq_cc_002-ext3}). $F_2\left( \Bigcdot \right)$ is the Appell Hypergeometric function which is defined in (\ref{eq_app_I-002-a-ext5}), Appendix I \cite{Saad2003}. The subindex $j$ in $\mathcal{W}_j(z)$ indicates different fading characteristics ($\kappa, \mu, m$) (or $\theta_1, \theta_2$) over each tier. 
\end{lemm}

\begin{proof}
	See Appendix III.
\end{proof}

By using a change of variable, \textit{i.e.}, $s z^{-1} = r^{\alpha}$, the Laplace transform of the interference can be expressed as $\mathcal{L}_{I}(s) = \exp\left[ -\sum_{j \in \mathcal{K}} 
	   	\pi \lambda_j \mathbb{E}\left[ \chi_j^{\delta} \right] \left( \hat{P}_j z^{-1}\right)^{\delta}	\mathcal{W}_j(z) s^{\delta}
	   	\right]$ which indicates that the aggregate interference is distributed by a Stable distribution as described below. Note that the exclusion zone in the interference field is considered in (\ref{eq_cc_012_ext1}) based on the condition $\|y_k^{\ast}\| = r$. 

\begin{lemm}
	The aggregate interference over a multiplicative channel of $\kappa$-$\mu$ shadowed fading and large-scale shadowing is distributed by a Stable distribution \cite{Haenggi2013} with four parameters; namely, stability $\delta$, skew = $1$, drift = $0$, dispersion = $\sec\left( \frac{\pi}{2} \delta\right)
		\sum_{j \in \mathcal{K}} 
	   	\pi \lambda_j \mathbb{E}\left[ \chi_j^{\delta} \right] \hat{P}_j^{\delta} z^{-\delta}\mathcal{W}_j(z)$ 
	with $\mathcal{W}_j(z)$ defined in (\ref{eq_cc_012_ext2}). The $j$-th order moment of the aggregate interference is given by 
	 \begin{equation}
	   \begin{split}
	   	\mathbb{E}\left[ I^j\right] =\frac{\Gamma\left( 1 - \frac{j}{\delta}\right)}{\Gamma\left( 1 - j\right) \cos\left( \frac{\pi}{2} \delta\right)^{\frac{j}{\delta}}}
	   	\left[ \sum_{j \in \mathcal{K}} 
	   	\pi r^2 \lambda_j \mathbb{E}\left[ \chi_j^{\delta} \right] z_j^{-\delta}
	   	\mathcal{W}(z_j) \right]^{\frac{\delta}{j}},
	   \end{split}
		\label{eq_cc_012_ext3}
	 \end{equation}	
	 for $j < \frac{2}{\alpha}$ and $\delta \neq 1$. Any moment with order above $j > \frac{2}{\alpha}$ is undefined, \textit{i.e.}, becomes infinity. 
\end{lemm}	   	

Using (\ref{eq_app_I-002-a-ext6}), the Appell's function reduces to a Gauss hypergeometric function if one of the parameters is zero. Hence the Laplace transform in (\ref{eq_cc_012_ext1}) and (\ref{eq_cc_012_ext2}) can be simplified as below. 

\begin{lemm}
	For the following fading distributions, $\mathcal{W}_j(z)$ in (\ref{eq_cc_012_ext2}) can be simplified as follows 
   \begin{align}
   	   &\text{Rayleigh}:~ &\frac{\bar{h} \delta z}{1-\delta} ~
        \Hypergeometric{2}{1}{1, 1-\delta}{2-\delta}{-\bar{h}z}
        \label{eq_cc_012_ext4}
        \\
       &\text{Nakagami-$m$}:~ &\frac{\bar{h}z}{1-\delta} ~
        \Hypergeometric{2}{1}{m+1, 1-\delta}{2-\delta}{-\frac{\bar{h}z}{m}} - \left[ 1 - {\left( 1 + \frac{\bar{h}z}{m} \right)^{-m}}\right]
        \label{eq_cc_012_ext5}
        \\
       &\text{One-Sided Gaussian}:~ &\frac{\bar{h}z}{1-\delta} ~
        \Hypergeometric{2}{1}{1.5, 1-\delta}{2-\delta}{-2\bar{h}z} - \left[ 1 - \frac{1}{\sqrt{
        1 + 2 \bar{h} z}} \right]
        \label{eq_cc_012_ext6}
        \\
       &\text{$\kappa$-$\mu$ fading}:~ &\frac{\mu \theta_1 z}{(1-\delta) \mathrm{e}^{\mu \kappa}} 
        \Hypergeometric{2}{1}{\mu+1, 1-\delta}{2-\delta}{-\theta_1 z} 
        - \left[ 1 - \frac{\mathrm{e}^{-\frac{\mu \kappa}{1 + (\theta_1 z)^{-1}}}  }
        {(1 + \theta_1 z)^{\mu}}
        \right] 
        \label{eq_cc_012_ext7}
        \\
        &\text{Rician}:~ &\frac{\theta_1 z}{(1-\delta) \mathrm{e}^{K}} 
        \Hypergeometric{2}{1}{2, 1-\delta}{2-\delta}{-\theta_1 z} 
        - \left[ 1 - \frac{\mathrm{e}^{-\frac{K}{1 + (\theta_1 z)^{-1}}}  }
        {1 + \theta_1 z}
        \right]
        \label{eq_cc_012_ext8}
   \end{align}
\end{lemm}	

\begin{proof}
	See Appendix IV.
\end{proof}

\begin{remark}
	The Appell's function in (\ref{eq_cc_012_ext2}) can be numerically evaluated by using \texttt{appellf2} function in SymPy package \cite{sympy}. Alternatively, we can use the Gauss-Laguerre Quadrature in (\ref{eq_app_I-002-a-ext9}) to approximate $\mathcal{W}_j(z)$ as follows
\begin{equation}
   \begin{split}
   \mathcal{W}_j(z) &= 
   \frac{\left( s \theta_1 r^{-\alpha} \right)^{\delta}}{\Gamma(\mu)} 
\left( \frac{\theta_1}{\theta_2}\right)^m
\int_{0}^{\infty}
t^{\delta+\mu-1}
\mathrm{e}^{-t}
\Hypergeometric{1}{1}{m}{\mu}{\frac{\mu \kappa}{\mu \kappa+m}t}
\gamma(1-\delta, s \theta_1 r^{-\alpha} t) 
\mathrm{d}t\\
&\quad - \left( 1 - (1 + \theta_1 z)^{m-\mu} (1 + \theta_2 z)^{-m}
\right)\\
		&= \frac{\left( \theta_1 z \right)^{\delta}}{\Gamma(\mu)} 
\left( \frac{\theta_1}{\theta_2}\right)^m
\sum_{n=1}^{N} w_n f(x_n) - \left( 1 - (1 + \theta_1 z)^{m-\mu} (1 + \theta_2 z)^{-m}
\right) + R_N,
   \end{split}
   \label{eq_cc_012_ext9}
 \end{equation}
where $x_n$ and $w_n$ are the $n$-th abscissa and weight of the $N$-th order Laguerre polynomial, $f(x) = x^{\delta+\mu-1} \Hypergeometric{1}{1}{m}{\mu}{\frac{\mu \kappa}{\mu \kappa+m}x}
\gamma(1-\delta,  \theta_1 z x)$, and $R_N$ is the residue term. Since $R_N$ converges rapidly to zero \cite{Gradshteyn1994}, (\ref{eq_cc_012_ext9}) provides a numerically accurate and efficient approximation to $\mathcal{W}_j(z)$. 
\end{remark}

\section{Theoretical Analysis of the Performance Measures}

In this section, we propose a novel method to compute $\mathbb{E}\left[ g\left( \gamma \right)\right]$ for an arbitrary function of the SINR $g(\gamma)$ using stochastic geometry. The original idea was proposed by Hamdi in \cite{Hamdi2007} for Nakagami-\textit{m} fading, and later in \cite{Chun2016} for $\kappa$-$\mu/\eta$-$\mu$ fading, which we further extend it to $\kappa$-$\mu$ shadowed fading in this paper and essentially to every linear fading model available in the open literature. By using the proposed method, one can evaluate any performance measures that are represented as a function of SINR (or SIR). For instance, the spectral efficiency, outage probability, moments of the SINR, and error probability can be expressed as an average of $g(x) = \log(1+x)$, $g(x) = \mathbb{I}(x \leq x_0)$, $g(x) = x^n$, and $g(x) = Q\left( x \right)$, respectively. 

  \begin{thm}
  For the $K$-tier HetNet with $\kappa$-$\mu$ shadowed fading, $\mathbb{E}\left[ g\left( \mathrm{SINR} \right)\right]$ is given by
    \begin{equation}
    \begin{split}
    &\mathbb{E}\left[ g\left( \mathrm{SINR} \right)\right] = \sum_{k = 1}^{K} \mathcal{P}_k \mathbb{E}\left[ g\left( \mathrm{SINR}_k \right)\right],\quad
      \mathbb{E}\left[ g\left( \mathrm{SINR}_k \right)\right] =
      \sum_{n=0}^{\infty} C_n \sum_{i=0}^{n} (-1)^{i} \binom{n}{i} \xi_i,
    \end{split}
    \label{eq_cc_023}
    \end{equation}
  where $\mathcal{P}_k$ is derived in (\ref{eq_cc_001-ext3}), $\mathrm{SINR}_k$ represents the SINR when a typical UE is associated to the $k$-th tier BS $y_k^{\ast}$, $C_n$ is defined in (\ref{eq_cc_012}), and $\xi_i$ represents the following integral
  \begin{equation}
    \begin{split}
    &\xi_i \triangleq \int_{0}^{\infty} g_{\mu+i}(z) ~
    \mathbb{E}_{r}\left[ \mathrm{e}^{-r^{\alpha} \hat{N} z} \mathcal{L}_{I}\left( r^{\alpha} z\right) \right] \mathrm{d}z,\\
    &\mathbb{E}_{r}\left[ \mathrm{e}^{-r^{\alpha} \hat{N} z} \mathcal{L}_{I}\left( r^{\alpha} z\right) \right] =
    \int_{0}^{\infty} \mathrm{e}^{-r^{\alpha} \hat{N} z} \mathcal{L}_{I}\left( r^{\alpha} z\right) f_{\|y_k^{\ast}\|}(r) \mathrm{d}r,
    \end{split}
    \label{eq_cc_024}
    \end{equation}
    the distribution $f_{\|y_k^{\ast}\|}(r)$ is given by (\ref{eq_cc_001-ext3}) and $\mathcal{L}_{I}(s)$ is derived in (\ref{eq_cc_012_ext1}). $g_{\mu+i}(z)$ is defined as
    \begin{equation}
    \begin{split}
      g_{\mu+i}(z) &= \frac{1}{\Gamma(\mu+i)} \frac{\mathrm{d}^{\mu+i}}{\mathrm{d}z^{\mu+i}}
      z^{\mu+i-1} g(z) = 
		\sum_{n=0}^{\mu+i-1}
		\binom{\mu+i}{n}
		\frac{z^{\mu+i-1-n}}{\Gamma(\mu+i-n)}
		\frac{\mathrm{d}^{\mu+i-n}}{\mathrm{d}x^{\mu+i-n}} g(z),
    \end{split}
    \label{eq_cc_025}
    \end{equation}
    where we used the general Leibniz rule in the last equality. 
  \end{thm}

  \begin{proof}
  See Appendix V.
  \end{proof} 

Theorem $1$ is the most general result in this paper that evaluates arbitrary performance measure for $K$-tier HetNet, considering noise, interference, per-tier BS density, and independent fading and shadowing across each tier. Theorem $1$ can be further simplified for some special cases, such as noise-limited scenario, interference-limited scenario, or identical fading and shadowing parameters on all tiers, which are described in Lemma 5. The analytic function $g(z)$ and the corresponding $g_{\mu+i}(z)$ for various performance measure are summarized in Table \ref{tab_gx}\footnote{The detailed proof of Table \ref{tab_gx} is given in \cite{Hamdi2007,Chun2016}.}.  We also note that $\mathbb{E}\left[ g\left( \mathrm{SINR}_k \right)\right]$ in (\ref{eq_cc_023}) is computationally efficient; the computational complexity of (\ref{eq_cc_023}) is same as a single summation expression since $\xi_i$ is independent of the index $n$.

\begin{lemm}
	$\mathbb{E}_{r}\left[ \mathrm{e}^{-r^{\alpha} \hat{N} z} \mathcal{L}_{I}\left( r^{\alpha} z\right) \right]$ can be evaluated for several path-loss exponent values as  
    \begin{equation}
    \mathbb{E}_{r}\left[ \mathrm{e}^{-r^{\alpha} \hat{N} z} \mathcal{L}_{I}\left( r^{\alpha} z\right) \right] = 
    \begin{dcases}
    \frac{\sum_{j \in \mathcal{K}} \lambda_j \mathbb{E}\left[ \chi_j^{\delta} \right] \hat{P}_j^{\delta} }{\sum_{j \in \mathcal{K}} \lambda_j \mathbb{E}\left[ \chi_j^{\delta} \right] \hat{P}_j^{\delta} \left( 1+\mathcal{W}_j(z) \right) + {\hat{N} z}/{\pi}}
    &\text{ for } \alpha = 2,\\
    \frac{\pi^{\frac{3}{2}} \sum_{j \in \mathcal{K}} \lambda_j \mathbb{E}\left[ \chi_j^{\delta} \right] \hat{P}_j^{\delta}}{2\sqrt{\hat{N} z}}
    \exp\left( \Theta^2 \right) \mathrm{erfc}\left( \Theta \right)
    &\text{ for } \alpha = 4.
	\end{dcases}
    \label{eq_cc_026-ext003}
    \end{equation}
    where $\Theta \triangleq \frac{\pi}{2\sqrt{\hat{N} z}} \sum_{j \in \mathcal{K}} \left( \lambda_j \mathbb{E}\left[ \chi_j^{\delta} \right] \hat{P}_j^{\delta} \left( 1+\mathcal{W}_j(z) \right) \right)$. If all tiers have identical fading characteristics ($\kappa, \mu, m$), then $\mathcal{W}(z) = \mathcal{W}_j(z)$ for any $j \in \mathcal{K}$ and (\ref{eq_cc_026-ext003}) can be further simplified as 
    \begin{equation}
    \begin{dcases}
    \frac{1}{1+\mathcal{W}(z) + \frac{\hat{N} z}{\pi \lambda_0}}
    &\text{ for } \alpha = 2,\\
    \frac{\pi^{\frac{3}{2}} \lambda_0}{2\sqrt{\hat{N} z}}
    \exp\left( \frac{\left(\pi \lambda_0 (1+\mathcal{W}(z))\right)^2}{4{\hat{N} z}} \right) \mathrm{erfc}\left( \frac{\pi \lambda_0 (1+\mathcal{W}(z))}{2\sqrt{\hat{N} z}} \right)
    &\text{ for } \alpha = 4.
	\end{dcases}
    \label{eq_cc_026-ext004}
    \end{equation}
    where we denoted $\lambda_0 \triangleq \sum_{j \in \mathcal{K}} \lambda_j \mathbb{E}\left[ \chi_j^{\delta} \right] \hat{P}_j^{\delta}$.

	For interference-limited scenario, \textit{i.e.}, $I \gg \hat{N}$, $\mathbb{E}_{r}\left[ \mathrm{e}^{-r^{\alpha} \hat{N} z} \mathcal{L}_{I}\left( r^{\alpha} z\right) \right]$ can be simplified as 
    \begin{equation}
    \int_{0}^{\infty} \mathcal{L}_{I}\left( r^{\alpha} z\right) f_{\|y_k^{\ast}\|}(r) \mathrm{d}r
    = \frac{\sum_{j \in \mathcal{K}} \lambda_j \mathbb{E}\left[ \chi_j^{\delta} \right] \hat{P}_j^{\delta} }{\sum_{j \in \mathcal{K}} \lambda_j \mathbb{E}\left[ \chi_j^{\delta} \right] \hat{P}_j^{\delta} \left( 1+\mathcal{W}_j(z) \right)}.
    \label{eq_cc_026-ext005}
    \end{equation}
    If all tiers have identical fading characteristics, then (\ref{eq_cc_026-ext005}) reduces to a succinct form as
    \begin{equation}
    \mathbb{E}_{r}\left[ \mathrm{e}^{-r^{\alpha} \hat{N} z} \mathcal{L}_{I}\left( r^{\alpha} z\right) \right] \overset{\hat{N} \rightarrow 0}{=\joinrel=} \frac{1}{1+\mathcal{W}(z)}. 
    \label{eq_cc_026-ext006}
    \end{equation}

    For noise-limited scenario, \textit{i.e.}, $I \ll \hat{N}$, $\mathbb{E}_{r}\left[ \mathrm{e}^{-r^{\alpha} \hat{N} z} \mathcal{L}_{I}\left( r^{\alpha} z\right) \right]$ can be simplified as 
    \begin{equation}
    	\begin{split}
    \int_{0}^{\infty} \mathrm{e}^{-r^{\alpha} \hat{N} z} f_{\|y_k^{\ast}\|}(r) \mathrm{d}r
     =
     \begin{dcases}
     \frac{1}{1+\frac{\hat{N} z}{\pi \lambda_0}},
     &\alpha = 2,\\
     \frac{\pi^{\frac{3}{2}} \lambda_0}{2\sqrt{\hat{N} z}}
    \exp\left( \frac{\left(\pi \lambda_0\right)^2}{4{\hat{N} z}} \right) \mathrm{erfc}\left( \frac{\pi \lambda_0}{2\sqrt{\hat{N} z}} \right),
    &\alpha = 4.
	 \end{dcases}    	
    	\end{split}
    \label{eq_cc_026-ext007}
    \end{equation}
\end{lemm}

\begin{proof}
 See Appendix VI.
\end{proof}

\begin{remark}
	If all tiers have identical fading characteristics and are interference-limited only, the performance measure $\mathbb{E}\left[ g\left( \mathrm{SINR} \right)\right]$ can be expressed by using Lemma $5$ as follows  
    \begin{equation}
    	\begin{split}
		\mathbb{E}\left[ g\left( \mathrm{SINR} \right)\right] =
		\sum_{n=0}^{\infty} C_n \sum_{i=0}^{n} (-1)^{i} \binom{n}{i} 
		\int_{0}^{\infty} \frac{g_{\mu+i}(z)}{\mathcal{W}(z)} \mathrm{d}z,
    \end{split}
    \label{eq_cc_027-ext001}
    \end{equation}
  where $C_n$, $g_{\mu+i}(z)$ and $\mathcal{W}(z)$ are independent of the PPP density $\lambda_j$. (\ref{eq_cc_027-ext001}) provides an important insight into the system performance of a PPP-distributed cellular network with $\kappa$-$\mu$ shadowed fading and arbitrary large-scale shadowing. Specifically, any performance measure of PPP-distributed HetNet that can be represented as a function of SINR is independent to the BS transmit power $P_k$, BS density $\lambda_k$, and the number of tiers $K$. This invariance property was originally introduced in \cite{Andrews2011,Jo2012,Dhillon2012} for Rayleigh fading. (\ref{eq_cc_027-ext001}) generalizes this argument by proving that the invariance property holds for any linear small-scale fading and finite large-scale shadowing distribution.
\end{remark}

\begin{remark}
If $\mu$ and $m$ are positive integers, (\ref{eq_cc_011_ext1}) can be utilized to achieve an expression analogous to Theorem 1, in terms of a single summation with finite terms as described below 
	 \begin{equation}
	   \begin{split}
	   \mathbb{E}\left[ g\left( \mathrm{SINR}_k \right)\right] = 
	   \begin{dcases}
	   \sum_{j=1}^{\mu-m} A_{1 j} \zeta_{\mu-m-j+1}(\theta_1) + 
	   \sum_{j=1}^{m} A_{2 j}  \zeta_{m-j+1}(\theta_2) &\text{for } m < \mu\\
	   \sum_{j=0}^{m-\mu} B_{j}  \zeta_{m-j}(\theta_2) &\text{for } m \geq \mu
	   \end{dcases},
	   \end{split}
		\label{eq_cc_028}
	\end{equation}
where $\zeta_j(\theta) = \int_{0}^{\infty} g_{j}(z) ~
    \mathbb{E}_{r}\left[ \mathrm{e}^{-\frac{r^{\alpha}}{\theta} \hat{N} z} \mathcal{L}_{I}\left( \frac{r^{\alpha}}{\theta} z\right) \right] \mathrm{d}z$
and the coefficients $A_{1 j}$, $A_{2 j}$, $B_{j}$ are derived in \cite[eq (6)]{Lopez-Martinez2016}. The proof of (\ref{eq_cc_028}) is omitted since it is analogous to Theorem 1. 
\end{remark}

In the following, we apply Theorem $1$ and Lemma $5$ to evaluate various performance measures.

\subsection{Spectral Efficiency}

	The spectral efficiency and the average user throughput are defined as \cite{Jo2012}
    \begin{equation}
    \begin{split}
    \mathcal{R} &= \sum_{k = 1}^{K} \mathcal{P}_k \mathbb{E}\left[ \ln\left( 1 + \mathrm{SINR}_k \right)\right], \quad 
    \bar{\mathcal{R}}_k = \frac{\mathbb{E}\left[ \ln\left( 1 + \mathrm{SINR}_k \right)\right]}{\mathcal{N}_k},
    \end{split}
    \label{eq_cc_027-ext002}
    \end{equation}
    where $\mathcal{P}_k$ is the tier association probability to the $k$-th tier evaluated by (\ref{eq_cc_002}), $\mathrm{SINR}_k$ is the received SINR from the $k$-th tier BS, and $\mathcal{N}_k$ represents the number of UEs served by the BS $x_k^{\ast}$\footnote{An accurate approximation to model the distribution of $\mathcal{N}_k$ is proposed in \cite{Singh2013}.}. The efficiency measures in (\ref{eq_cc_027-ext002}) require $\mathbb{E}\left[ \ln\left( 1 + \mathrm{SINR}_k \right)\right]$ which can be evaluated by using Theorem $1$ with $g(z) = \ln(1+z)$ and $g_{\mu+i}(z)$ as follows \cite{Hamdi2007} 
    \begin{equation}
    \begin{split}
      g_{\mu+i}(z) &= \frac{1}{\Gamma(\mu+i)} \frac{\mathrm{d}^{\mu+i}}{\mathrm{d}z^{\mu+i}}
      z^{\mu+i-1} g(z) = \frac{1}{z} \left( 1 - \frac{1}{(1+z)^{\mu+i}}\right).
    \end{split}
    \label{eq_cc_027}
    \end{equation}

    Given identical channel characteristics across each tier, the spectral efficiency reduces to
    \begin{equation}
    \begin{split}
    \mathcal{R} = \mathbb{E}\left[ \ln\left( 1 + \mathrm{SINR}_k \right)\right]
    = \sum_{n=0}^{\infty} C_n \sum_{i=0}^{n} (-1)^{i} \binom{n}{i} 
		\int_{0}^{\infty} \frac{\mathcal{K}(z)}{z}\left( 1 - \frac{1}{(1+z)^{\mu+i}}\right) \mathrm{d}z,    
	\end{split}
    \label{eq_cc_027-ext004}
    \end{equation}
    where $C_n$ is derived in (\ref{eq_cc_010}), $\mathcal{W}(z)$ is defined as (\ref{eq_cc_012_ext2}) and 
    $\mathcal{K}(z)$ is given by 
    \begin{equation}
    \begin{split}
    \mathcal{K}(z)
    =  
     \begin{dcases}
		\frac{1}{1+\mathcal{W}(z) + \frac{\hat{N} z}{\pi \lambda_0}},
     &\text{for }\alpha = 2,\\
		\frac{\sqrt{\pi} \Theta}{1 + \mathcal{W}(z)} \exp\left(\Theta^2\right) \mathrm{erfc}\left( \Theta \right),
    &\text{for }\alpha = 4,\\
		\frac{1}{1+\mathcal{W}(z)},
    &\text{for interference-limited environments}\\    
	 \end{dcases}.
    \end{split}
    \label{eq_cc_027-ext005}
    \end{equation}
    By substituting (\ref{eq_cc_027-ext004}) to (\ref{eq_cc_027-ext002}), the throughput can be derived. 

\subsection{Moments of the SINR}

	Higher order moments of the SINR are a crucial performance measure which have an important role in the determination of network performance. $\mathbb{E}\left[ \mathrm{SINR}^r \right]$ can be evaluated by using Theorem $1$ with $g(z) = z^r$ and $g_{\mu+i}(z)$ as 
    \begin{equation}
    \begin{split}
      g_{\mu+i}(z) &= \frac{1}{\Gamma(\mu+i)} \frac{\mathrm{d}^{\mu+i}}{\mathrm{d}z^{\mu+i}} z^{\mu+i-1} g(z) 
      = \frac{\Gamma(\mu + i + r)}{\Gamma(r) \Gamma(\mu + i)} z^{r-1}.
    \end{split}
    \label{eq_cc_027-ext006}
    \end{equation}	
    For the case when we have identical channel characteristics across each tier, the $r$-th order moment is simplified to
    \begin{equation}
    \begin{split}
    \mathbb{E}\left[ \mathrm{SINR}^r \right] = 
    \sum_{n=0}^{\infty} C_n \sum_{i=0}^{n} (-1)^{i} \binom{n}{i} \frac{\Gamma(\mu + i + r)}{\Gamma(r) \Gamma(\mu + i)}
		\int_{0}^{\infty} z^{r-1} \mathcal{K}(z) \mathrm{d}z,    
	\end{split}
    \label{eq_cc_027-ext007}
    \end{equation}
    where $\mathcal{K}(z)$ is defined in (\ref{eq_cc_027-ext005}). Using moments of the SINR, the moment generating function (MGF) of the SINR can be obtained as follows 
	\begin{equation}
    \begin{split}
    \mathcal{M}_{\mathrm{SINR}}(t) = 
    \mathbb{E}\left[ \mathrm{e}^{t \times \mathrm{SINR}} \right] = 
    \sum_{r=0}^{\infty} \frac{t^r}{r!} \mathbb{E}\left[ \mathrm{SINR}^r \right].    
	\end{split}
    \label{eq_cc_027-ext008}
    \end{equation}

\subsection{Outage Probability and Rate Coverage Probability}

The outage probability and rate coverage probability are defined as 
    \begin{equation}
    \begin{split}
    &P_o(T_o) =   \mathbb{P}\left( \mathrm{SINR} < T_o \right) = \sum_{k = 1}^{K} \mathcal{P}_k \mathbb{P}\left( \mathrm{SINR}_k < T_o \right),\quad
    R_c = \mathbb{P}\left( \mathcal{R} > T_r \right), 
	\end{split}
    \label{eq_cc_027-ext011}
    \end{equation}
respectively for a predefined SINR threshold $T_o$ and rate threshold $T_r$. Theoretically, one can use Theorem $1$ to calculate (\ref{eq_cc_027-ext011}) by approximating the step function with a smooth sigmoid function, \textit{i.e.}, $g(z) = \mathbb{I}(z < T_o) \simeq \frac{1}{1+\mathrm{e}^{-\epsilon(z - T_o)}}$, where $\epsilon$ controls the sharpness. However, even with a smooth function, $g_{\mu+i}(z)$ behaves like an impulse signal for a large derivation order $\mu+i$ \cite{Minai1993}. Hence, most numerical software will present a precision overflow while evaluating (\ref{eq_cc_024}).

Instead of using Theorem $1$, it appears more convenient to use the two-step method based on Campbell's theorem \cite{Haenggi2013,Andrews2011} for the outage and rate coverage probability analysis as follows:\\ \textit{\textbf{Step 1)}} 
the conditional SINR distributions $\mathbb{P}\left( \mathrm{SINR}_k < T_o \right)$ in (\ref{eq_cc_027-ext011}) can be evaluated as 
    \begin{equation}
    \begin{split}
&\int_{0}^{\infty} \mathbb{P}\left( \mathrm{SINR}_k < T_o | ~\|y_k^{\ast}\| = r \right) f_{\|y_k^{\ast}\|}(r) \mathrm{d}r
    = \int_{0}^{\infty} \mathbb{P}\left( h _{y_k^{\ast}} < T r^{\alpha} (I + \hat{N}) \right)
    f(r) \mathrm{d}r\\
    = &\sum_{n=0}^{\infty} \sum_{i=0}^{n} b_{i, n} \int_{0}^{\infty} \mathbb{E}_{t}\left[ t^{\mu+i} \mathrm{e}^{-t}\right] f(r) \mathrm{d}r + 
	\sum_{n = 0}^{\infty} \frac{1}{\Gamma(\mu + n + 1)} \int_{0}^{\infty} \mathbb{E}_{t}\left[ t^{\mu+n} \mathrm{e}^{-t}\right] f(r) \mathrm{d}r,
	\end{split}
    \label{eq_cc_027-ext013}
    \end{equation}
where we used (\ref{eq_cc_002}) in the second equality and $t = T r^{\alpha} (I + \hat{N})$, (\ref{eq_cc_011}), (\ref{eq_app_I-002-a-ext4}) in the last equality. The tier association probability $\mathcal{P}_k$ and the PDF of the link length $f_{\|y_k^{\ast}\|}(r)$ are derived in (\ref{eq_cc_001-ext3}).
\textit{\textbf{Step 2)}} The term $\mathbb{E}_{t}\left[ t^{n} \mathrm{e}^{-t}\right]$ in (\ref{eq_cc_027-ext013}) can be evaluated as follows
    \begin{equation}
   \begin{split}
   &\mathbb{E}_{t}\left[ t^{n} \mathrm{e}^{-t} \right] = (-1)^n \left. \frac{\partial^n \mathcal{L}_{t}(s)}{\partial s^n}\right|_{s = 1},\quad
   \mathcal{L}_{t}(s) = \mathbb{E}\left[ \mathrm{e}^{-s T r^{\alpha} (I + \hat{N})}\right] = \mathrm{e}^{-s T r^{\alpha} \hat{N}} \mathcal{L}_{I}\left( s T r^{\alpha} \right),
   \end{split}
   \label{eq_cc_027-ext014}
 \end{equation}
 where $\mathcal{L}_{I}\left( s \right)$ is derived in (\ref{eq_cc_012_ext1}).
Based on the Leibniz rule, we can interchange the order of the integral and derivative as follows 
    \begin{equation}
    \begin{split}
	\int_{0}^{\infty} \mathbb{E}_{t}\left[ t^{n} \mathrm{e}^{-t}\right] f_{\|y_k^{\ast}\|}(r) \mathrm{d}r = (-1)^n \left. \frac{\partial^n}{\partial s^n}
	\int_{0}^{\infty}
	\mathrm{e}^{-s T r^{\alpha} \hat{N}} \mathcal{L}_{I}\left( s T r^{\alpha} \right) f_{\|y_k^{\ast}\|}(r) \mathrm{d}r
	\right|_{s = 1}.
	\end{split}
    \label{eq_cc_027-ext015}
    \end{equation}
Again assuming identical channel characteristics across each tier, (\ref{eq_cc_027-ext015}) can be simplified as 
    \begin{equation}
    \begin{split}
	\int_{0}^{\infty} \mathbb{E}_{t}\left[ t^{n} \mathrm{e}^{-t}\right] f_{\|y_k^{\ast}\|}(r) \mathrm{d}r = (-1)^n \left. \frac{\partial^n }{\partial s^n} \mathcal{K}\left( s T_o\right)
	\right|_{s = 1} = (-1)^n \left. \mathcal{K}^{(n)}\left( s T_o\right)
	\right|_{s = 1},
	\end{split}
    \label{eq_cc_027-ext016}
    \end{equation}
and the outage probability of a $K$-tier HetNet can be expressed in a succinct form as below
    \begin{equation}
    \begin{split}
    P_o(T_o) = \sum_{n=0}^{\infty} \sum_{i=0}^{n} b_{i, n} 
	(-1)^n \left. \mathcal{K}^{(\mu+i)}\left( s T_o\right)
	\right|_{s = 1} + 
	\sum_{n = 0}^{\infty} \frac{(-1)^n \left. \mathcal{K}^{(\mu+n)}\left( s T_o\right)
	\right|_{s = 1}}{\Gamma(\mu + n + 1)},	
	\end{split}
    \label{eq_cc_027-ext017}
    \end{equation}
where $\mathcal{K}(z)$ and $b_{i, n}$ are defined in (\ref{eq_cc_027-ext005}) and (\ref{eq_cc_012}), respectively\footnote{The $n$-th order derivatives in  (\ref{eq_cc_027-ext017}) can be numerically evaluated by using Faa di Bruno's formula \cite{Huang2006}, which is a well-known and widely accepted technique to calculate the interference functional \cite{Tanbourgi2014,Dhillon2013}}.

\section{Numerical Results}

In this section, we present numerical evaluations of the theoretical results and compare them with Monte-Carlo simulations. All of the numerical results presented in this paper were obtained by using the Julia language which provides fast computation time and easy syntax  that is similar to Python and Matlab \cite{Bezanson2015}. In our analysis, we considered a two-tier HetNet with BS intensity $\lambda_1 = 2 \lambda_2$, transmit power $P_2 = P_1 - 20$ dB, a path-loss exponent $\alpha = 4$ and lognormal distributed $\chi$ with $\mu_l = 0$ dB. Without the loss of generality, we assumed identical fading and shadowing parameters for both tiers. \figref{fig:fig2} (a)-(c) compare the spectral efficiency versus the channel parameter $\kappa$ for lognormal, gamma, and inverse-Gaussian distributed large-scale shadowing coefficient $\chi$. For a small $\sigma$, we note that the gamma and inverse-Gaussian distributed shadowing accurately approximate the rate performance of a link which experiences lognormal distributed shadowing. However, given a large $\sigma$, there is a notable discrepancy between the rate of the lognormal distribution and the others. We also observe that the rate performance gap between lognormal, gamma, and inverse-Gaussian distribution becomes wider as the $m$ parameter of the $\kappa$-$\mu$ shadowed fading decreases, and vice versa. 

\figref{fig:fig1} compares the spectral efficiency and average SINR across a wide range of channel parameters $(\kappa, \mu, m)$. In Figs \ref{fig:fig1} (a)-(f), we considered an interference-limited environment where the aggregate interference power is larger than the noise power. We note that a strong dominant LOS component (large $\kappa$) and rich scattering (large $\mu$) collectively achieve a higher rate. However, the average rate decreases on a weak shadowing condition (large $m$), which may at first seem counter-intuitive. Small $m$ indicates a strong random fluctuation on the dominant component, which decreases not only the received signal power but also the aggregate interference power, increasing the SINR level, and eventually achieving higher spectral efficiency. In contrast, given a large $m$, random fluctuation of the dominant component subsides and $\kappa$-$\mu$ shadowed fading reduces to $\kappa$-$\mu$  fading, which increases the interference power, deteriorating the received SINR level as well as the average rate. 

Figs \ref{fig:fig1} (g)-(i) compare the spectral efficiency versus the macro BS intensity $\lambda_1$ without interference-limited conditions. As conjectured in Remark 2, the spectral efficiency becomes invariant for a large BS intensity $\lambda_1$. In a dense network with a large BS intensity, the aggregate interference becomes significantly larger than the noise power, achieving an interference-limited condition. Additionally we observe that the BS intensity required to reach the rate asymptote is inversely proportional to the operating SNR level. For a high SNR regime, the average rate reaches the asymptote around $\lambda_1 = 10^{-2}$, whereas in a low SNR regime, a large number of BSs ($\lambda_1 \geq 10^{-1}$) are required to obtain sufficiently larger interference power than the noise. 

Similarly, Fig \ref{fig:fig1} (j)-(l) compare the average SINR versus the SNR for various BS intensities $\lambda_1$, where we calculated the SNR at a unit distance, \textit{i.e.}, $\text{SNR} = \frac{\mathbb{E}\left[\chi\right] \bar{h}}{\hat{N}}$ with $\|x_k^{\ast}\| = 1$. For a dense network ($\lambda_1 \geq 10^{-2}$), the aggregate interference surpasses the noise power even at low SNR levels, resulting in an interference-limited environment with a constant average SINR. This is in contrast to a sparse network ($\lambda_1 \leq 10^{-4}$), where each BS can increase their transmit power even further than the dense network without saturating the average SINR. Nonetheless, the average SINR level of a sparse network is much lower than that of a dense network. For example the average SINR is about $0.2$ for $\lambda_1 = 10^{-4}$ and $2.0$ for $\lambda_1 = 10^{-2}$ with $\kappa = 6$ and $\text{SNR} = 15$ dB. 


\section{Conclusion}

In this paper, we have considered a cellular network in which the signal fluctuation is the result of large-scale and LOS shadowing to encapsulate the diverse range of channel conditions that can occur in 5G communications. We applied a Laguerre polynomial series expansion to represent the $\kappa$-$\mu$ shadowed fading distribution as a simplified series expression. Based on the series expressions, we then proposed a novel stochastic geometric method to evaluate the average of an arbitrary function of the SINR over $\kappa$-$\mu$ shadowed fading channels. The proposed method is numerically efficient, can be easily applied to other network models, and can evaluate any performance measure that can be represented as a function of SINR. Using the proposed method, we have evaluated the spectral efficiency, moments of the SINR, bit error probability and outage probability of a $K$-tier HetNet with $K$ classes of BSs, differing in terms of the transmit power, BS density, shadowing characteristics and small-scale fading. Finally, we provided numerical results and investigated the performance over a range of channel parameters and observed that a dominant LOS component (large $\kappa$), rich scattering environment (large $\mu$) and strong shadowing condition (small $m$) collectively provides high spectral efficiency.

The analytical framework proposed in this paper can be applied to practical use cases of 5G communications, where Rayleigh fading fails to fully capture the diverse nature of the underlying channel environment. The effect of diverse channel conditions on the second order interference statistics is also an important measure, which needs to be studied to optimize the network performance. Furthermore, the proposed framework can be extended to multi-slope pathloss model, which will provide an accurate channel model for practical communications. 

\section*{Acknowledgment}
The work of Y. J. Chun, S. L. Cotton and S. K. Yoo was supported in part by the Engineering and Physical Sciences Research Council (EPSRC) under Grant References EP/L026074/1, and in part by the Dept. for Employment and Learning Northern Ireland through Grant No. USI080. The work of H. S. Dhillon was supported by US NSF (Grants CCF-1464293, CNS-1617896, IIS-1633363).

\section*{Appendix I}

In this appendix, we summarize the operational equalities of the special functions, which are used in this paper\footnote{Most of the expressions in Appendix I were introduced in \cite{Gradshteyn1994}, except for (\ref{eq_app_I-002-a-ext5}) and (\ref{eq_app_I-002-a-ext6}), which were proved in \cite{Saad2003}.}. First, the generalized Laguerre polynomial of degree $n$ and order $\beta$ has the following functional identities
        \begin{align}
		&L_{n}^{\beta}(t) = \sum_{i = 0}^{n} (-1)^i \binom{n+\beta}{n-i} \frac{t^i}{i!},
		\label{eq_app_I-001-a}
		\\
		&t^{\beta} \exp\left( -t \right) L_{n}^{\beta}(t) \mathrm{d}t = \frac{1}{n}
		\mathrm{d}\left[ t^{\beta+1} \exp\left( -t \right) L_{n-1}^{\beta+1}(t) \right].
        \label{eq_app_I-001-b}
        \end{align}
The following properties of hypergeometric function hold for real constants $a, b$ and $c$
        \begin{align}
        &\Hypergeometric{1}{1}{a}{b}{t} = \mathrm{e}^{t} \Hypergeometric{1}{1}{b-a}{b}{-t}, \quad 
		\Hypergeometric{2}{1}{a, b}{c}{z} = (1-z)^{-a} \Hypergeometric{2}{1}{a, c-b}{c}{\frac{z}{z-1}},
		\label{eq_app_I-002-a-ext1}
		\\
        &\int_{0}^{\infty} t^{\alpha-1} \mathrm{e}^{-c t} \Hypergeometric{1}{1}{a}{b}{-t} \mathrm{d}t =
   c^{-\alpha} \Gamma(\alpha) \Hypergeometric{2}{1}{a, \alpha}{b}{-\frac{1}{c}}, \quad \alpha > 0 ~\mathrm{and}~ c > 0,
        \label{eq_app_I-002-b}
        \\
		&
		\resizebox{0.9\textwidth}{!}{$
		\left( (a-b) z + c - 2 a \right) \Hypergeometric{2}{1}{a, b}{c}{z} = 
		\left( c-a \right) \Hypergeometric{2}{1}{a-1, b}{c}{z} + 
		a \left( z-1 \right) \Hypergeometric{2}{1}{a+1, b}{c}{z}$},
		\label{eq_app_I-002-a-ext2}
		\\
		&
		\int_{0}^{\infty} \mathrm{e}^{-(a x^2 + b x)}\mathrm{d}x = 
		\frac{1}{2}\sqrt{\frac{\pi}{\alpha}} \exp\left( \frac{b^2}{4 a}\right) 
		\mathrm{erfc}\left( \frac{b}{2 \sqrt{a}} \right), \quad a > 0 ~\mathrm{and}~ b > 0.
		\label{eq_app_I-002-b-ext1}		
        \end{align}
The lower incomplete gamma function $\gamma(s, x) = \int_{0}^{x} t^{\mu-1} \mathrm{e}^{-t} \mathrm{d}t$ has the following series representation and functional identity for arbitrary positive real constant $s$
        \begin{align}
		&\frac{\gamma(s, x)}{\Gamma(s)} = \sum_{n = 0}^{\infty}
		\frac{x^{s+n} \mathrm{e}^{-x}}{\Gamma(s + n + 1)}, \qquad 
		\gamma(s, x) = s^{-1} x^s \mathrm{e}^{-x} \Hypergeometric{1}{1}{1}{1+s}{x}.
		\label{eq_app_I-002-a-ext4}
        \end{align}

The binomial coefficient can be defined for real constants $x, y$ using the gamma function as 
        \begin{align}
		\binom{x}{y} = \frac{\Gamma(x+1)}{\Gamma(y+1)\Gamma(x-y+1)},\quad
		\Gamma(t) = \int_{0}^{\infty}x^{t-1} \mathrm{e}^{-x} \mathrm{d}x.
        \label{eq_app_I-003}
        \end{align}
Appell's function $F_2\left( \Bigcdot \right)$ is defined via the Pochhammer symbol $(x)_{n} = \frac{\Gamma(x+n)}{\Gamma(x)}$ as follows
        \begin{align}
		&F_2\left( \alpha; \beta, \beta^{'}; \gamma, \gamma^{'}; x, y \right) = \sum_{m=0}^{\infty} \sum_{n=0}^{\infty}
		\frac{\left( \alpha \right)_{m+n} \left( \beta \right)_{m} \left( \beta^{'} \right)_{n}}{m!~ n! \left( \gamma \right)_{m} \left( \gamma^{'} \right)_{n}} x^m y^n.
		\label{eq_app_I-002-a-ext5}
        \end{align}
Appell's function can be reduced to the hypergeometric function using the following properties
  \begin{equation}
  \begin{split}
	F_2\left( d; a, a^{'}; c, c^{'}; 0, y \right) &= \Hypergeometric{2}{1}{d, a^{'}}{c^{'}}{y},\quad
	F_2\left( d; a, a^{'}; c, c^{'}; x, 0 \right) = \Hypergeometric{2}{1}{d, a}{c}{x}.
  \end{split}
  \label{eq_app_I-002-a-ext6}
  \end{equation}
The following integration holds under the following constraints $d > 0$ and $|k| + |k|^{'} < |h|$  
\begin{equation}
  \begin{split}
  \int_{0}^{\infty} t^{d-1} \mathrm{e}^{-h t} \Hypergeometric{1}{1}{a}{b}{k t} \Hypergeometric{1}{1}{a^{'}}{b^{'}}{k^{'} t} \mathrm{d}t = h^{-d} \Gamma(d) F_2\left( d; a, a^{'}; b, b^{'}; \frac{k}{h}, \frac{k^{'}}{h} \right).
  \end{split}
  \label{eq_app_I-002-a-ext7}
  \end{equation}
Gaussian quadratures can be used to evaluate the following integral for a given analytic function $g(x)$ as 
        \begin{align}
		\text{Gauss-Laguerre Quadrature}; &\quad 
		\int_{0}^{\infty} \mathrm{e}^{-x} g(x) \mathrm{d}x = \sum_{n=1}^{N} w_n f\left( x_n \right) + R_N,
		\label{eq_app_I-002-a-ext9}	        
        \end{align}
  where $x_n$ and $w_n$ are the $n$-th abscissa and weight of the $N$-th order Laguerre polynomial.

\section*{Appendix II}

In this appendix, we provide a proof of Lemma $1$. The PDF of $h$ for $\kappa$-$\mu$ shadowed fading in (\ref{eq_cc_005}) can be represented in the orthogonal series expansion form as
  \begin{equation}
  \begin{split}
	f_h(x) = \sum_{n=0}^{\infty} C_n \left(
	\frac{n! ~ L_n^{\mu-1}(x)}{\Gamma(n+\mu)}
	\right) x^{\mu-1} \exp(-x), \quad 0 \leq x < \infty,
  \end{split}
  \label{eq_app_II-001}
  \end{equation}
where we applied the Laguerre polynomial series expansion in \cite[eq.9]{Chai2009} and the coefficient $C_n$ is evaluated by substituting (\ref{eq_cc_005}) as follows \cite[eq.8]{Chai2009}
\begin{equation}
	\begin{split}
	C_n &= \int_{0}^{\infty} L_n^{\mu-1}(x) f_h(x) \mathrm{d}x\\
   &=
 \frac{\theta_1^{m-\mu} }{\theta_2^m \Gamma(\mu)}
  \underbrace{
 \int_{0}^{\infty}
 x^{\mu-1}  \exp\left( -\frac{x}{\theta_1}\right) L_n^{\mu-1}(x)
 \Hypergeometric{1}{1}{m}{\mu}{\frac{\theta_2 - \theta_1}{\theta_1 \theta_2}x}\mathrm{d}x}_{\mathrm{I}_1}.
 \end{split}
	\label{eq_app_II-002}
\end{equation}
The integral $\mathrm{I}_1$ can be simplified by using the series representation of $L_n^{\mu-1}(x)$ in (\ref{eq_app_I-001-a}) as follows
\begin{equation}
	\begin{split}
	\mathrm{I}_1 
	 &= \sum_{i=0}^{n} \frac{(-1)^i}{i!} \binom{n+\mu-1}{n-i}
	 \frac{\theta_2^m \Gamma(\mu)}{\theta_1^{m-\mu}} \mathbb{E}\left[ h^i \right],
	\end{split}
	\label{eq_app_II-003}
\end{equation}
where we used (\ref{eq_cc_005}) to express the integral as the PDF of the $\kappa$-$\mu$ shadowed fading in the last equality. Then, by substituting (\ref{eq_app_II-003}) into (\ref{eq_app_II-002}), the coefficient $C_n$ in (\ref{eq_cc_012}) can be derived after algebraic manipulation. The series expansion form in (\ref{eq_app_II-001}) can be further simplified by using (\ref{eq_app_I-001-a}) and (\ref{eq_app_I-003}) as follows
  \begin{equation}
  \begin{split}
	f_h(x) &= x^{\mu-1} \exp(-x)
	\sum_{n=0}^{\infty}
	\frac{n!~C_n}{\Gamma(n+\mu)}
	\left(	\sum_{i=0}^{n}(-1)^{\mu-1} \binom{n+\mu-1}{n-i}\frac{x^i}{i!}	\right)
	\\
	&=
	\sum_{n=0}^{\infty}  \sum_{i=0}^{n}
	\frac{(-1)^i ~ C_n}{\Gamma(\mu+i)} \binom{n}{i} x^{\mu+i-1} \exp(-x)
  \end{split}
  \label{eq_app_II-004}
  \end{equation}
which achieves (\ref{eq_cc_010}).

The CDF of $h$ can be evaluated as follows
  \begin{equation}
  \begin{split}
	F_h(x) &= \int_{0}^{x} f_h(t) \mathrm{d}t =
	\sum_{n=0}^{\infty}
	\frac{n! ~ C_n}{\Gamma(n+\mu)}
	\int_{0}^{x}
	t^{\mu-1} \exp(-t) L_n^{\mu-1}(t) \mathrm{d}t\\
	&=\sum_{n=1}^{\infty}	 \frac{\Gamma(n) C_n }{ \Gamma(n+\mu)}
	x^{\mu} \exp(-x) L_{n-1}^{\mu}(x) + 
  \frac{C_0}{\Gamma(\mu)}
  \int_{0}^{x}
  t^{\mu-1} \exp(-t) L_0^{\mu-1}(t) \mathrm{d}t \\
	&= \sum_{n=0}^{\infty} \sum_{i=0}^{n} \frac{(-1)^i \Gamma(n+1) C_{n+1}}{i!~ \Gamma(n+\mu+1)}
	\binom{n+\mu}{n-i} x^{\mu+i} \exp(-x) + \frac{\gamma(\mu, x)}{\Gamma(\mu)},
  \end{split}
  \label{eq_app_II-005}
  \end{equation}
where we used (\ref{eq_app_II-001}) in the second equality, utilized (\ref{eq_app_I-001-b}) in the third equality, applied a change of variable, \textit{i.e.}, $n^{'} \leftarrow n-1$, $C_0 = 1$, $L_0^{\mu-1}(t) = 1$ and (\ref{eq_app_I-001-a}) in the last equality. The coefficient $b_{i, n}$ can be simplified by using (\ref{eq_app_I-003}) as
  \begin{equation}
  \begin{split}
b_{i, n} = \frac{(-1)^i \Gamma(n+1) C_{n+1}}{i!~ \Gamma(n+\mu+1)}
	\binom{n+\mu}{n-i} = \frac{(-1)^i C_{n+1}}{\Gamma(\mu+i+1)}
	\binom{n}{i},
  \end{split}
  \label{eq_app_II-006}
  \end{equation}
then the CDF in (\ref{eq_cc_011}) can be subsequently obtained. This completes the proof.

  \section*{Appendix III}

  In this appendix, we provide a proof of Lemma $2$. Due to (\ref{eq_cc_001-ext2}), all interfering BS within the $j$-th tier are located further than $\hat{P}_j^{\frac{1}{\alpha}} \|y_k^{\ast}\|$ where $y_k^{\ast}$ denote the associated $k$-th tier BS and $\hat{P}_j = \frac{P_j}{P_k}$ is the transmit power ratio between the interfering and serving BS
  \begin{equation}
  \begin{split}
  &P_j \|y\|^{-\alpha} < P_k \|y_k^{\ast}\|^{-\alpha} ~\text{for any } y \in \Phi_j^{(e)}\backslash\{ y_k^{\ast}\} \quad 
  \leftrightarrow \quad \|y\| > \hat{P}_j^{\frac{1}{\alpha}} \|y_k^{\ast}\|.
  \end{split}
  \label{eq_app_III-001-ext1}
  \end{equation}
  The Laplace transform of the interference from the $j$-th tier is given by 
  \begin{equation}
  \begin{split}
  \mathcal{L}_{I_j}(s) &= \mathbb{E}\left[ \exp(-s 
 	   	\sum_{y \in \Phi_j^{(e)}\backslash\{ y_k^{\ast}\}}  \hat{P}_j h_{y} \|y\|^{-\alpha}
 	   	)\right]\\ 
 	   	&= \exp\left[ -2\pi \lambda_j \mathbb{E}[\chi_j^{\delta}]
 	   	\int_{\hat{P}_j^{\frac{1}{\alpha}} r}^{\infty} \left( 
 	   	1 - \mathbb{E}_h\left[\exp\left( -s \hat{P}_j h l^{-\alpha} \right) \right]
 	   	\right) l \mathrm{d}l
 	   	\right]\\
 	   	&= \exp\left[ -\pi \lambda_j \mathbb{E}[\chi_j^{\delta}] \hat{P}_j^{\delta}~
 	   	\mathbb{E}_h\left\{ 
 	   	\left( s h \right)^{\delta}
 	   	\int_{0}^{s h r^{-\alpha}}
 	   	\delta t^{-\delta-1} (1-\mathrm{e}^{-t}) \mathrm{d}t
 	   	\right\}
 	   	\right]\\ 	   	
 	   	&= \exp\left[ -\pi r^2 \lambda_j \mathbb{E}[\chi_j^{\delta}] \hat{P}_j^{\delta}~
 	   	\mathbb{E}_h\left\{ 
 	   	\left( s h r^{-\alpha} \right)^{\delta}
 	   	\gamma(1-\delta, s h r^{-\alpha})
 	   	- \left( 1 - \mathrm{e}^{-s h r^{-\alpha}} \right)
 	   	\right\}
 	   	\right], 	   	
  \end{split}
  \label{eq_app_III-001-ext2}
  \end{equation}
where we represented the distance to the serving BS as $\|y_k^{\ast}\| = r$ in the second equality, applied a change of variable, \textit{i.e.}, $s \hat{P}_j h l^{-\alpha} = t$, in the third equality, then used integration by parts. 

The first part of the expectation term in (\ref{eq_app_III-001-ext2}) can be evaluated as follows
  \begin{equation}
  \begin{split}
 	   	&\mathbb{E}_h\left[
 	   	\left( s h r^{-\alpha} \right)^{\delta}
 	   	\gamma(1-\delta, s h r^{-\alpha})
 	   	\right] \\
= &\frac{\left( s \theta_1 r^{-\alpha} \right)^{\delta}}{\Gamma(\mu)} 
\left( \frac{\theta_1}{\theta_2}\right)^m
\int_{0}^{\infty}
t^{\delta+\mu-1}
\mathrm{e}^{-t}
\Hypergeometric{1}{1}{m}{\mu}{\frac{\mu \kappa}{\mu \kappa+m}t}
\gamma(1-\delta, s \theta_1 r^{-\alpha} t) 
\mathrm{d}t\\
= &\frac{s \theta_1 r^{-\alpha} \left( \theta_1/\theta_2\right)^m}{(1-\delta)\Gamma(\mu)} 
\int_{0}^{\infty}
t^{\mu}
\mathrm{e}^{-(1+s\theta_1 r^{-\alpha})t}
\Hypergeometric{1}{1}{m}{\mu}{\frac{\mu \kappa}{\mu \kappa+m}t}
\Hypergeometric{1}{1}{1}{2-\delta}{s \theta_1 r^{-\alpha} t}
\mathrm{d}t \\
= &\frac{\mu}{(1-\delta)} \frac{s\theta_1 r^{-\alpha}}{(1+s\theta_1 r^{-\alpha})^{\mu+1}} 
\left( \frac{\theta_1}{\theta_2}\right)^m
F_2\left( \mu+1; m, 1; \mu, 2-\delta; A, B\right)
  \end{split}
  \label{eq_app_III-001-ext3}
  \end{equation}
where we used the PDF of $\kappa$-$\mu$ shadowed fading with a change of variable, \textit{i.e.}, $\frac{h}{\theta_1} = t$ in the first equality, applied (\ref{eq_app_I-002-a-ext4}) to the second equality, utilized the integration (\ref{eq_app_I-002-a-ext7}) in the last equality \cite{Saad2003}, $A = \frac{1-\theta_1/\theta_2}{1+\theta_1 s r^{-\alpha}}$ and $B = \frac{\theta_1 s r^{-\alpha}}{1+\theta_1 s r^{-\alpha}}$. The second part of the expectation term in (\ref{eq_app_III-001-ext2}) follows directly by using the Laplace transform of $\kappa$-$\mu$ shadowed channel coefficient (\ref{eq_cc_005}). By denoting $s r^{-\alpha} = z$, (\ref{eq_cc_012_ext1}) and (\ref{eq_cc_012_ext2}) can be achieved. This completes the proof.

  \section*{Appendix IV}

  In this appendix, we provide a proof of Lemma $4$. First, we consider Nakagami-\textit{m} fading which corresponds to the case when $\kappa \rightarrow 0, \mu = m$ in Table 1. Then $\theta_1 = \theta_2 = \frac{\bar{h}}{m}$ and $A = \frac{1-\theta_1/\theta_2}{1+\theta_1 z} \rightarrow 0$. By applying (\ref{eq_app_I-002-a-ext6}) and (\ref{eq_app_I-002-a-ext1}), (\ref{eq_cc_012_ext2}) can be simplified to the following form 
	\begin{equation}
	   \begin{split}
	   F_2\left( \mu+1; m, 1; \mu, 2-\delta; A, B \right) = \left( 1 + \theta_1 z\right)^{\mu+1} \Hypergeometric{2}{1}{\mu+1, 1-\delta}{2-\delta}{-\theta_1 z}.
	   \end{split}
		\label{eq_app_IV-001-ext1}
	 \end{equation}
$\mathcal{W}_j(z)$ for Nakagami-\textit{m} fading can be obtained by substituting (\ref{eq_app_IV-001-ext1}) in (\ref{eq_cc_012_ext2}) together with $\kappa \rightarrow 0, \mu = m$. For One-sided Gaussian fading, (\ref{eq_cc_012_ext6}) is obtained by substituting $\mu = 0.5$ in (\ref{eq_cc_012_ext5}). $\mathcal{W}_j(z)$ for Rayleigh fading in (\ref{eq_cc_012_ext4}) can be obtained by substituting $\mu = 1$ in (\ref{eq_cc_012_ext5}), then applying (\ref{eq_app_I-002-a-ext2}) and $\Hypergeometric{2}{1}{0, b}{c}{x} = 1$, which achieves an identical result to \cite[eq. (44)]{Jo2012}. 

Next, we show that $\kappa$-$\mu$ fading corresponds to the case of $m \rightarrow \infty$ with the following limit 
	\begin{equation}
	   \begin{split}
			\lim_{m \rightarrow \infty} \left( \frac{\theta_1}{\theta_2}\right)^m &= 
			\lim_{m \rightarrow \infty} \left( 1 + \frac{\mu \kappa}{m}\right)^{-m} = \mathrm{e}^{-\mu \kappa}\\
			\lim_{m \rightarrow \infty} \left( \frac{1+\theta_1 s}{1+\theta_2 s}\right)^m 
			&= 
			\lim_{m \rightarrow \infty} \left( 
			1 + \frac{\mu \kappa s}{m (s+\theta_1^{-1})}\right)^{-m} = 
			\exp\left(-\frac{\mu \kappa s}{s+\theta_1^{-1}}\right).
	   \end{split}
		\label{eq_app_IV-001-ext3}
	 \end{equation}
By utilizing (\ref{eq_app_IV-001-ext3}) and (\ref{eq_app_IV-001-ext1}) in (\ref{eq_cc_012_ext2}), 
(\ref{eq_cc_012_ext7}) can be derived for $m \rightarrow \infty$. $\mathcal{W}_j(z)$ for Rician fading readily follows by substituting $\kappa = K$ and $\mu = 1$ in (\ref{eq_cc_012_ext7}). This completes the proof.

  \section*{Appendix V}

  In this appendix, we provide a proof of Theorem $1$.
  The average of an arbitrary function of the SINR $\frac{h_{x_0}  \|x_0\|^{-\alpha}}{I+N}$ is written as follows
    \begin{equation}
    \begin{split}
      &\mathbb{E}\left[ \left. g\left( \frac{h_{x_0}  r^{-\alpha}}{I+N} \right) \right\vert I,  \|x_0\| = r \right]
      = \int_{0}^{\infty} g\left( \frac{x r^{-\alpha}}{I+N}\right) f_{h}(x) \mathrm{d}x\\
      &\qquad= \sum_{n=0}^{\infty} \sum_{i=0}^{n} c_{i, n}
      \int_{0}^{\infty}
      ~x^{\mu+i-1} \mathrm{e}^{-x}
      g\left( \frac{x r^{-\alpha}}{I+N}\right) \mathrm{d}x\\
      &\qquad=
      \sum_{n=0}^{\infty} \sum_{i=0}^{n} c_{i, n}
      \int_{0}^{\infty}
      z^{\mu+i-1} g(z) (r^{\alpha} (I+N))^{\mu+i} \mathrm{e}^{-r^{\alpha} (I+N) z} \mathrm{d}z\\
      &\qquad=
      \sum_{n=0}^{\infty} C_n \sum_{i=0}^{n} (-1)^i \binom{n}{i}
      \int_{0}^{\infty}
      \frac{z^{\mu+i-1} g(z)}{\Gamma(\mu + i)}
     (r^{\alpha} (I+N))^{\mu+i} \mathrm{e}^{-r^{\alpha} (I+N) z} \mathrm{d}z,
    \end{split}
    \label{eq_app_VIII-001}
    \end{equation}
    where (\ref{eq_cc_010}) is applied in the second equality, a change of variable, \textit{i.e.},
    $\frac{x r^{-\alpha}}{I+N} = z$, is utilized in the third equality, and (\ref{eq_cc_012}) is employed in the last equality. (\ref{eq_app_VIII-001}) can be evaluated as follows
    \begin{equation}
    \begin{split}
    \int_{0}^{\infty} \underbrace{\frac{z^{\mu+i-1}}{\Gamma(\mu+i)} g\left( z \right)}_{u} ~
      \underbrace{
      \vphantom{\frac{z^{\mu+i-1}}{\Gamma(\mu+i)} g\left( z \right)}
      b^{\mu + i} \mathrm{e}^{-b z}}_{v'} \mathrm{d}x
      &=
      \left. -\sum_{k = 0}^{\mu+i-1} g_k(z) b^{\mu+i-k-1} \mathrm{e}^{-bz} \right\vert_{0}^{\infty} + \int_{0}^{\infty}
      g_{\mu+i}(z)\mathrm{e}^{-bz} \mathrm{d}z,
    \end{split}
    \label{eq_app_VIII-004}
    \end{equation}
  where we denoted $b = r^{\alpha} (I+N)$, applied integration by parts, defined $g_k(z)$ in (\ref{eq_cc_025}), and
    \begin{equation}
      g_k(0) =
      \begin{dcases}
      0, & \text{for } k < \mu+i-1\\
      g(0), & \text{for } k = \mu+i-1
      \end{dcases}.
    \label{eq_app_VIII-005}
    \end{equation}
  Then, the average of an arbitrary function of the SINR is given by
    \begin{equation}
    \begin{split}
      &\mathbb{E}\left[ g\left( \frac{h_{x_0}  \|x_0\|^{-\alpha}}{I+N} \right) \right] =
      \mathbb{E}\left[ \mathbb{E}\left[ \left. g\left( \frac{h_{x_0}  r^{-\alpha}}{I+N} \right)
      \right\vert \|x_0\| = r \right]\right]\\
      &=\sum_{n=0}^{\infty} C_n \sum_{i=0}^{n} (-1)^i \binom{n}{i}
      \int_{0}^{\infty} g_{\mu+i}(z)
      \mathrm{e}^{-r^{\alpha} N z} \mathcal{L}_{I}\left( r^{\alpha} z\right) f_{\|x_0\|}(r) \mathrm{d}r
    \end{split}
    \label{eq_app_VIII-006}
    \end{equation}
  where we used $\sum_{i=0}^{n} (-1)^i \binom{n}{i} = 0$ in the second equality. This completes the proof.

  \section*{Appendix VI}
	In this appendix, we provide a proof of Lemma $5$. By substituting (\ref{eq_cc_001-ext3}) and (\ref{eq_cc_012_ext1}) to (\ref{eq_cc_024}), the expectation term $\mathbb{E}_{r}\left[ \mathrm{e}^{-r^{\alpha} \hat{N} z} \mathcal{L}_{I}\left( r^{\alpha} z\right) \right]$ can be evaluated as follows
    \begin{equation}
    	\begin{split}
    	&\mathbb{E}_{r}\left[ \mathrm{e}^{-r^{\alpha} \hat{N} z} \mathcal{L}_{I}\left( r^{\alpha} z\right) \right] = 
    	\int_{0}^{\infty} \mathrm{e}^{-r^{\alpha} \hat{N} z} \mathcal{L}_{I}\left( r^{\alpha} z\right) f_{\|y_k^{\ast}\|}(r) \mathrm{d}r\\
    	= 
		&\frac{2 \pi \lambda_k \mathbb{E}\left[ \chi_k^{\delta}\right]}{\mathcal{P}_k} \int_{0}^{\infty} r ~\mathrm{e}^{-r^{\alpha} \hat{N} z}
		    \exp\left[ 
		    -\sum_{j \in \mathcal{K}} \pi r^2 \lambda_j \mathbb{E}\left[ \chi_j^{\delta} \right] 
		    \hat{P}_j^{\delta} \left( 1 +  \mathcal{W}_j(z) \right) 
		\right] 
		\mathrm{d}r\\
		= & \frac{\lambda_k \mathbb{E}\left[ \chi_k^{\delta}\right]}{\mathcal{P}_k} \int_{0}^{\infty} 
		    \exp\left[ 
		    - t^{\frac{1}{\delta}} \frac{\hat{N} z}{\pi^{\frac{\alpha}{2}}} 
		    - t \left( \sum_{j \in \mathcal{K}} \lambda_j \mathbb{E}\left[ \chi_j^{\delta} \right] 
		    \hat{P}_j^{\delta} \left( 1 +  \mathcal{W}_j(z) \right)
		    \right)
		\right] 
		\mathrm{d}r,  
    	\end{split}
    \label{eq_cc_026-ext008}
    \end{equation}
    where we used a change of variable, \textit{i.e.}, $t = \pi r^2$ in the last equality. If $\alpha = 2$, then (\ref{eq_cc_026-ext008}) becomes
    \begin{equation}
    	\begin{split}
	&\frac{\lambda_k \mathbb{E}\left[ \chi_k^{\delta}\right]}{\mathcal{P}_k} \int_{0}^{\infty} 
		    \mathrm{e}^{
		    - t \left( \frac{\hat{N} z}{\pi}  + \sum_{j \in \mathcal{K}} \lambda_j \mathbb{E}\left[ \chi_j^{\delta} \right] 
		    \hat{P}_j^{\delta} \left( 1 +  \mathcal{W}_j(z) \right)
		    \right)
		}
		\mathrm{d}r\\
	= &\frac{\lambda_k \mathbb{E}\left[ \chi_k^{\delta}\right]/\mathcal{P}_k}{\sum_{j \in \mathcal{K}} \lambda_j \mathbb{E}\left[ \chi_j^{\delta} \right] 
		    \hat{P}_j^{\delta} \left( 1 +  \mathcal{W}_j(z) \right) 
		    + \frac{\hat{N} z}{\pi}},   
    	\end{split}
    \label{eq_cc_026-ext009}
    \end{equation}
		which achieves (\ref{eq_cc_026-ext003}) by substituting $\mathcal{P}_k$ from (\ref{eq_cc_001-ext3}). Similarly, $\mathbb{E}_{r}\left[ \mathrm{e}^{-r^{\alpha} \hat{N} z} \mathcal{L}_{I}\left( r^{\alpha} z\right) \right]$ for $\alpha = 4$ can be obtained by applying (\ref{eq_app_I-002-b-ext1}). 
	Given an interference-limited condition, (\ref{eq_cc_026-ext008}) reduces to 
    \begin{equation}
    	\begin{split}
		\frac{\lambda_k \mathbb{E}\left[ \chi_k^{\delta}\right]}{\mathcal{P}_k} \int_{0}^{\infty} 
		    \mathrm{e}^{
		     - t \left( \sum_{j \in \mathcal{K}} \lambda_j \mathbb{E}\left[ \chi_j^{\delta} \right] 
		    \hat{P}_j^{\delta} \left( 1 +  \mathcal{W}_j(z) \right)
		    \right)
		} 
		\mathrm{d}r = \frac{\lambda_k \mathbb{E}\left[ \chi_k^{\delta}\right]/\mathcal{P}_k }{\sum_{j \in \mathcal{K}} \lambda_j \mathbb{E}\left[ \chi_j^{\delta} \right] \hat{P}_j^{\delta} \left( 1+\mathcal{W}_j(z) \right)},
    	\end{split}
    \label{eq_cc_026-ext010}
    \end{equation}
	whereas for noise-limited condition, (\ref{eq_cc_026-ext008}) can be written as 
    \begin{equation}
    	\begin{split}
	\frac{\lambda_k \mathbb{E}\left[ \chi_k^{\delta}\right]}{\mathcal{P}_k} \int_{0}^{\infty} 
		    \exp\left[ 
		    - t^{\frac{1}{\delta}} \frac{\hat{N} z}{\pi^{\frac{\alpha}{2}}} 
		    - t \left( \sum_{j \in \mathcal{K}} \lambda_j \mathbb{E}\left[ \chi_j^{\delta} \right] 
		    \hat{P}_j^{\delta} \right)
		\right] 
		\mathrm{d}r.
    	\end{split}
    \label{eq_cc_026-ext011}
    \end{equation}
	(\ref{eq_cc_026-ext007}) readily follows by substituting $\mathcal{W}_j(z) \rightarrow 0$ in (\ref{eq_cc_026-ext003}). This completes the proof.

\bibliographystyle{IEEEtran}
\bibliography{bib1}

\clearpage 

\begin{table}[!t]
    \centering
    \caption{Special Cases of the $\kappa$-$\mu$ Shadowed Fading Model.}
    \begin{tabular}{| l|c |c |c |}
    \hline
    ~ & $\kappa$-$\mu$ fading & $\eta$-$\mu$ fading & $\kappa$-$\mu$ shadowed fading\\
    \hline
    Rayleigh & $\kappa \rightarrow 0, \mu = 1$ & $\eta = 1, \mu = 0.5$ &
    \begin{tabular}{@{}c@{}}
    $\kappa \rightarrow 0, \mu = 1$ or \\
    $m = 1, \mu = 1$
    \end{tabular}
    \\
    \hline
    Nakagami-\textit{m} & $\kappa \rightarrow 0, \mu = m$ &
    \begin{tabular}{@{}c@{}}
    $\eta = 1, \mu = m/2$ or \\
    $\eta \rightarrow 0, \mu = m$
    \end{tabular}
    &
    \begin{tabular}{@{}c@{}}
    $\kappa \rightarrow 0, \mu = m$ or \\
    $m \rightarrow m, \mu = m$
    \end{tabular} \\
    \hline
    Nakagami-\textit{n} (Rice) & $\mu = 1$ & ~ & $\kappa = K, \mu = 1, m \rightarrow \infty$ \\
    \hline
    Nakagami-\textit{q} (Hoyt) & ~ & $\mu = 0.5$ & $\kappa = (1-q^2)/2q^2, \mu = 1, m = 0.5$ \\
    \hline
    One-sided Gaussian & $\kappa \rightarrow 0, \mu = 0.5$ &
    \begin{tabular}{@{}c@{}}
    $\eta \rightarrow 0, \mu = 0.5$ or \\
    $\eta \rightarrow \infty, \mu = 0.5$
    \end{tabular}
    &
    \begin{tabular}{@{}c@{}}
    $\kappa \rightarrow 0, \mu = 0.5$ or \\
    $m = 0.5, \mu = 0.5$
    \end{tabular}
    \\
    \hline
    $\kappa$-$\mu$ fading& $\kappa, \mu$ & ~ & $\kappa \rightarrow \kappa, \mu \rightarrow \mu, m \rightarrow \infty$ \\
    \hline
    $\eta$-$\mu$ fading& ~ & $\eta, \mu$ & $\kappa = (1-\eta)/2\eta, \mu \rightarrow 2\mu, m = \mu$\\
    \hline
    Rician shadowed & ~ & ~ & $\kappa = K, \mu = 1, m = m$\\
    \hline
    \end{tabular}
    \label{tab_channel1}
\end{table}

\begin{table}
\centering
  \caption{Different $g(x)$ and $g_{\mu+i}(x)$ for evaluating various system measures.}
  \begin{tabular}{|l|l|l|}
  \hline
  Measure                            & $g(x)$           & $g_{\mu+i}(x) = \frac{1}{\Gamma(\mu+i)} \frac{\mathrm{d}^{\mu+i}}{\mathrm{d}z^{\mu+i}}
      x^{\mu+i-1} g(x)$ \\ \hline
  Rate                               & $\log(1+x)$       & $\frac{1}{x} \left( 1 - \frac{1}{(1+x)^{\mu+i}}\right)$\\ \hline
  Higher order moments                 & $x^r$       &$\frac{\Gamma(\mu+i+r)}{\Gamma(\mu+i) \Gamma(r)} x^{r-1}$
  \\ \hline
  Outage probability                 & $\mathbb{I}(x \leq x_0) \simeq \frac{1}{1+\mathrm{e}^{-\epsilon(x - x_0)}}$      
                                     &  
                                     $\sum_{k=0}^{\mu+i-1}
    \binom{\mu+i}{k}
    \frac{z^{\mu+i-1-k}}{\Gamma(\mu+i-k)}
    \frac{\mathrm{d}^{\mu+i-k}}{\mathrm{d}x^{\mu+i-k}} \mathbb{I}(x \leq x_0)$
                                                   \\ \hline
  \end{tabular}
    \label{tab_gx}
\end{table}

\begin{figure}[!t]
  \centering
    \includegraphics[width=0.8\linewidth]{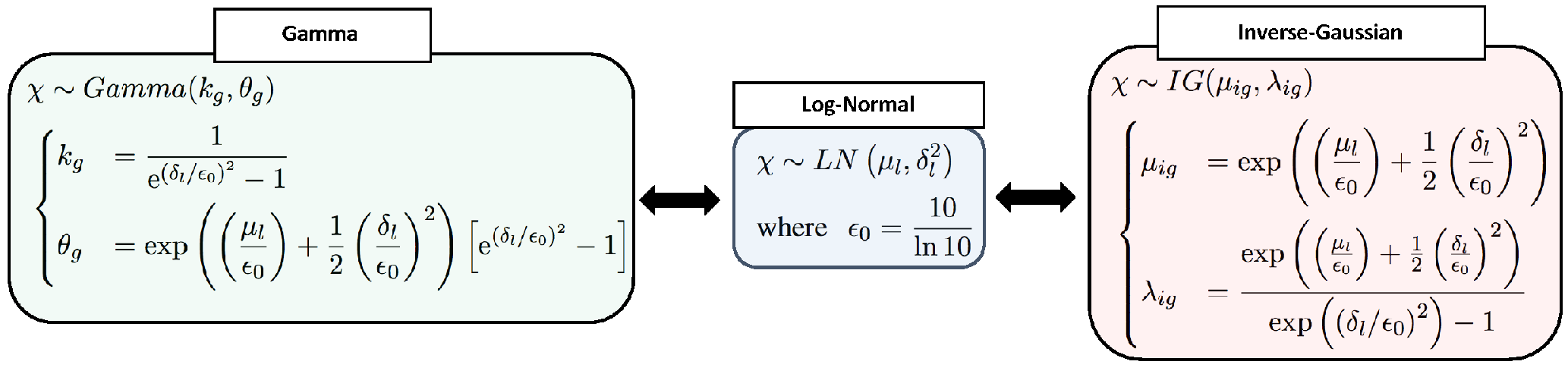}%
    \caption{One-to-one mapping between lognormal, gamma and Inverse-Gaussian shadowing based on the moment matching.}
    \label{fig:shadowing_parameter}
\end{figure}

\begin{figure}[!t]
  \centering
    \includegraphics[height=0.7\myimageheight]{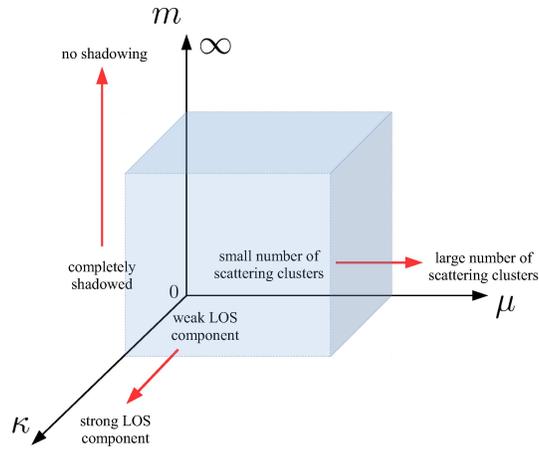}%
    \caption{Physical meaning of the channel parameters $(\kappa, \mu, m)$ in $\kappa$-$\mu$ shadowed fading model.}
    \label{fig:channel_parameter}
\end{figure}

\begin{figure}[!t]
  \centering
    \includegraphics[height=0.7\myimageheight]{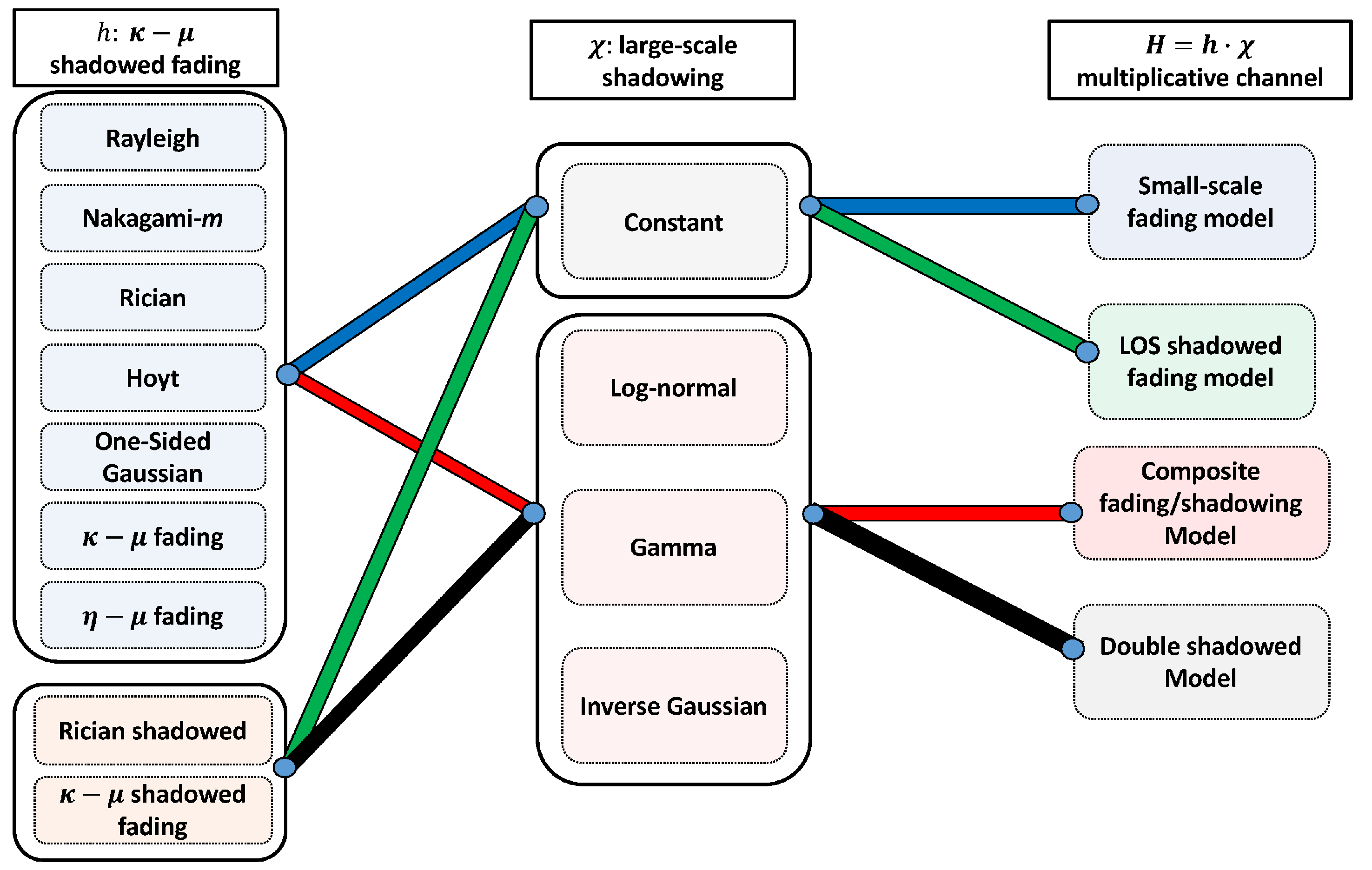}%
    \caption{Versatility of the proposed channel model $H = h \cdot \chi$ with $\kappa$-$\mu$ shadowed fading  $h$ and large-scale shadowing $\chi$.}
    \label{fig:channel_parameter2}
\end{figure}

\begin{figure}[!t]
	\makebox[\textwidth][c]{\includegraphics[width=1.15\textwidth]{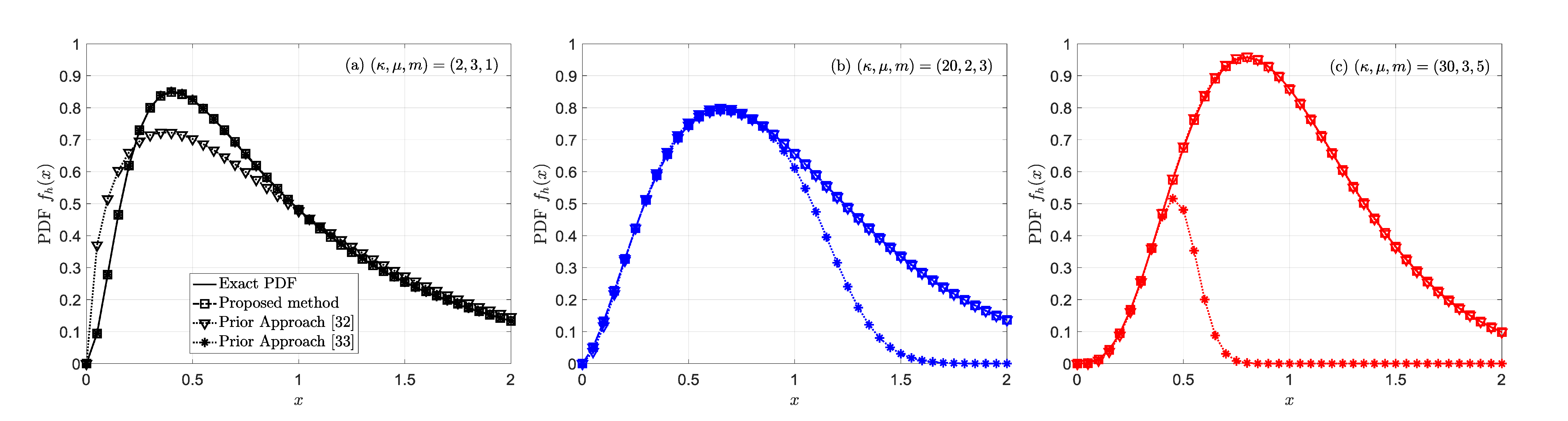}}%
    \caption{Numerical evaluation of the $\kappa$-$\mu$ shadowed fading distribution using the PDFs given in (\ref{eq_cc_005}), (\ref{eq_cc_010}), \cite{Kumar2015}, and \cite{Parthasarathy2016}.}
    \label{fig:pdf_compare}
\end{figure}

\begin{figure}[!t]
	\makebox[\textwidth][c]{\includegraphics[width=1.15\textwidth]{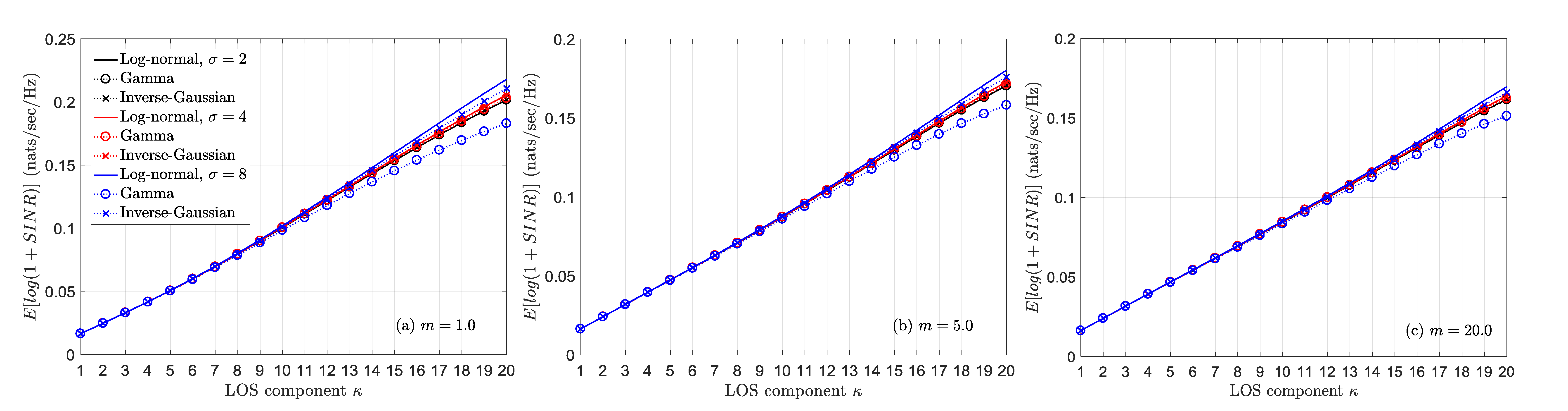}}%
	\caption{Spectral efficiency versus the channel parameter $\kappa$ for different large-scale shadowing distribution $\chi$.}
	\label{fig:fig2}
\end{figure}

\begin{figure}[!t]
	\makebox[\textwidth][c]{\includegraphics[width=1.15\textwidth]{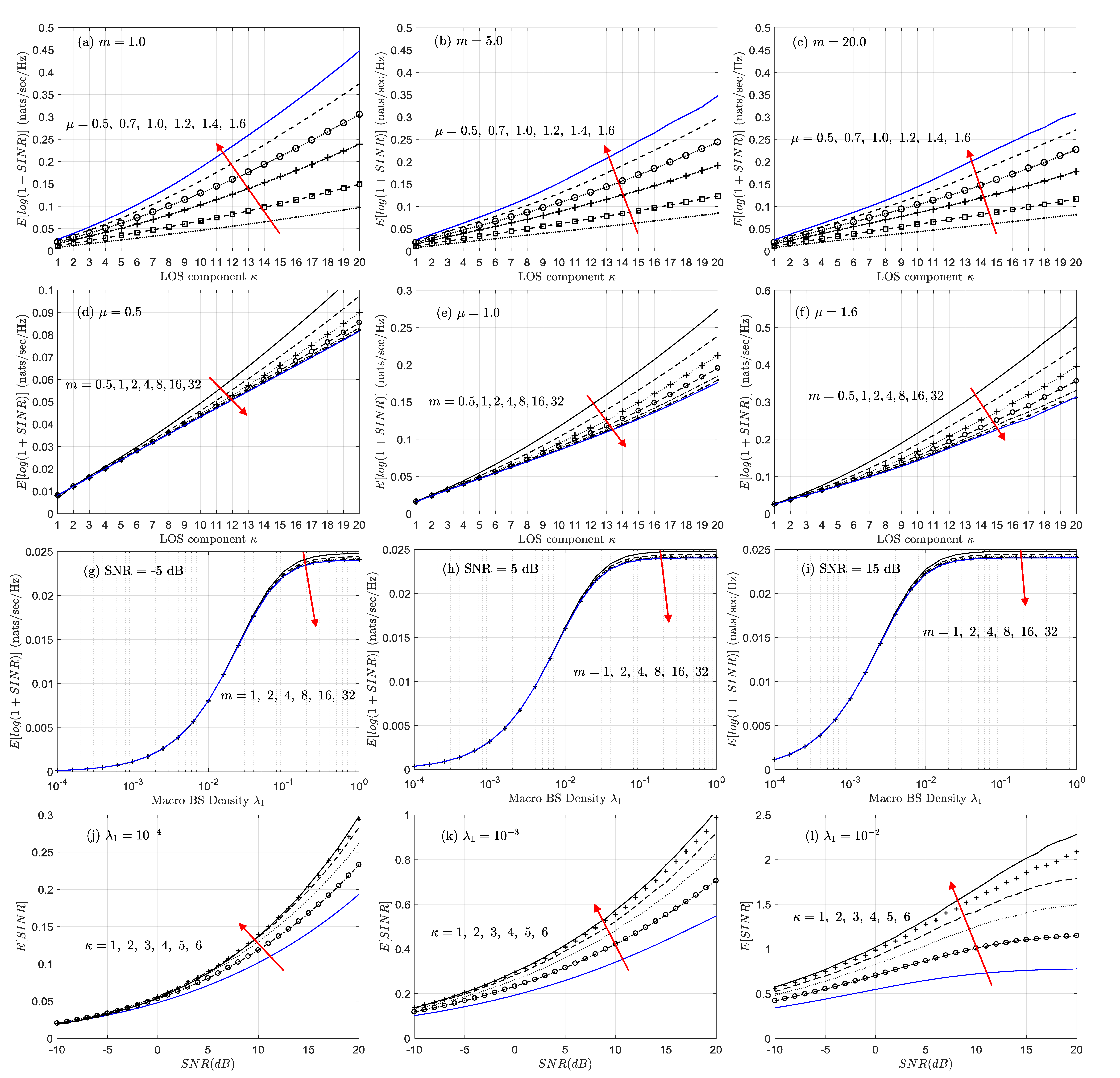}}%
	\caption{Spectral efficiency and average SINR of a two-tier HetNet over various channel parameters $(\kappa, \mu, m)$; (a)-(f) assume interference-limited environment with $\lambda_1 = \frac{1}{\pi 500^2}$, $P_1 = 53$ dBm, $\bar{h} = 1$ and $\sigma _l = 4$ dB. (g)-(i) assume $\kappa = 2$, $\mu = 1$, $\bar{h} = 1$ and $\sigma _l = 4$ dB, whereas (j)-(l) assume $\kappa = 1$, $\mu = 1$, $\bar{h} = 1$ and $\sigma _l = 4$ dB.}
	\label{fig:fig1}
\end{figure}

\end{document}